\RequirePackage{amsmath}
\documentclass{llncs}

\usepackage[T1]{fontenc}
\usepackage{amsfonts}
\usepackage{amssymb}
\usepackage{amsmath}
\usepackage{tikz}
\usetikzlibrary{automata}
\usetikzlibrary{shadows}
\usetikzlibrary{arrows}
\usetikzlibrary{shapes}
\usetikzlibrary{decorations.pathmorphing}
\usetikzlibrary{fit}
\usepackage{algorithm}
\usepackage{algorithmic}
\usepackage[margin=2cm]{geometry}
\usepackage{dsfont}
\usepackage{enumitem}

\def\slash{\relax\ifmmode\delimiter"502F30E\mathopen{}\else\@old@slash\fi}

\tikzstyle{every picture}=[
  >=stealth', shorten >=1pt, node distance=1.44cm,auto,bend angle=45,initial text=,
  every state/.style={inner sep=0.75mm, minimum size=1mm},font=\small,
]

\def\sc{\mathrm{sc}}

\begin{document}

\title{State complexity of multiple catenation}
\author{
    Pascal Caron
    \and Jean-Gabriel Luque
    \and Bruno Patrou
    %\thanks{LITIS, Universit\'e de Rouen, 76801 Saint-\'Etienne du Rouvray Cedex, France}
    \thanks{\{Pascal.Caron, Jean-Gabriel.Luque, Bruno.Patrou\}@univ-rouen.fr}
}
\institute{Département d'Informatique, Université de Rouen,\\ Avenue de l'Université,\\ 76801 Saint-\'Etienne du Rouvray Cedex,\\ France}
%\email{\{Pascal.Caron, Jean-Gabriel.luque, Bruno.Patrou\}@univ-rouen.fr}

%\tableofcontents
\maketitle
\begin{abstract}
We improve some results relative to the state complexity of the multiple catenation described by Gao and Yu. %  in \cite{GY09}. 
In particular we nearly divide by $2$ the size of the alphabet needed for witnesses. We also give some refinements to the algebraic expression of the state complexity, which is especially complex with this operation. We obtain these results by using peculiar DFAs defined by Brzozowski. %in \cite{Brz13}.
\end{abstract}

% !TEX root =  main.tex
\section{Introduction}

State complexity is a very active research area. It aims to determine the maximal size of a minimal automaton recognizing a language belonging to a given class. State complexity can be studied from the deterministic as well as non-deterministic point of view. Here, we only consider the deterministic case. Then, the state complexity of a regular language is the states number of its minimal DFA (Deterministic Finite Automaton). And the state complexity of a regular operation allows to compute the maximal size of any DFA obtained by applying this operation over regular languages, knowing their respective state complexities. Such operations can be elementary (see, as one of the first reference in this domain, \cite{YZS94}) or the result of some combinations (see, for example, \cite{GSY08}, \cite{CGKY11} or \cite{JO11}). Sometimes, the computation of state complexities needs to use heavy tools of combinatorial, as in \cite{CLMP15}. To have an expanded view of the domain, it is useful to refer to the surveys \cite{GMRY12} and \cite{GMRY15}.

In \cite{YZS94}, the authors are the first ones to study the state complexity of catenation. They prove $m2^n-2^{n-1}$ to be the upper bound for the states number of a minimal DFA recognizing the catenation of two regular languages with respective state complexities $m$ and $n$. And they propose a $3$-letters witness reaching the bound. In \cite{Jir05}, G. Jiraskova produces a $2$-letters witness. In \cite{GY09}, the authors study a generalization by considering the sequential catenation of an arbitrary number $\alpha$ of regular languages. The upper bound they find is very intricate to write, its algebraic representation being growing with $\alpha$. The witnesses they describe are defined over $(2\alpha-1)$-letters alphabets. In \cite{Brz13}, J. Brzozowski shows that a particular family of DFAs can be used to produce witnesses in a very large number of cases.

In this paper, we focus on sequential catenation of $\alpha$ DFAs and our contributions are the following: first, we give a recursive definition of the state complexity which can be easily computed. Then, as our main result, we improve the set of witnesses by dramatically reducing the size of the alphabet from $2\alpha-1$ to $\alpha+1$. For this, we use DFAs issued from the Brzozowski family. Last, we conjecture it is possible to decrease the size of the alphabet until $\alpha$ (which should be optimal) still using Brzozowski DFAs. We test computationally our conjecture until $6$ or $7$ DFAs, and prove it when $\alpha=2$ (giving here a positive issue to a remark made by Brzozowski who thought its family was deficient in this peculiar case) and $\alpha=3$.

In section $2$ are recalled the classical tools we need both in automata theory and in algebraic combinatorics. Section $3$ is devoted to the presentation of the construction used for multiple catenation and to compute the upper bound for the state complexity of this construction. In section $4$, we describe a family of $\alpha$ DFAs over an $(\alpha+1)$-letters alphabet and prove it to be a witness for the catenation of $\alpha$ regular languages. For the same operation, we give, in section $5$, witnesses over $\alpha$-letters alphabet when $\alpha=2$ and $\alpha=3$ and we conjecture these witnesses can be extended for any value of $\alpha$.
% !TEX root =  main.tex
\section{Preliminaries}

In all this paper, $\Sigma$ 
%lud suppr will 
denotes a finite alphabet. The set of all finite words over $\Sigma$ is denoted by $\Sigma ^*$.  The empty word is denoted by  $\varepsilon$. A language is a subset of $\Sigma^*$. The set of subsets of a finite set $A$ is denoted by $2^A$ and $\#A$ denotes the cardinality of $A$.  %The symmetric difference is denoted by the symbol $\oplus$. We denote by $\uplus$ the union of disjoint sets. The symbol $\circ$ denotes any boolean operation on languages. 
In the following, by abuse of notation, we 
%lud suppr will 
often write $q$ for any singleton $\{q\}$.

A  finite automaton (FA) is a $5$-tuple $A=(\Sigma,Q,I,F,\cdot)$ where $\Sigma$ is the input alphabet, $Q$ is a finite set of states, $I\subset Q$ is the set of initial states, $F\subset Q$ is the set of final states and $\cdot$ is the transition function from  $Q\times \Sigma$ to $2^Q$. A FA is deterministic (DFA) if $\#I=1$ and for all $q\in Q$, for all $a\in \Sigma$, $\#(q\cdot a)\leq 1$. 
 Let $a$ be a symbol of $\Sigma$. Let $w$ be a word of $\Sigma^*$. The transition function is  extended to any word  by $q\cdot a w=\bigcup _{q'\in q\cdot a} q'\cdot w$ and  $q\cdot \varepsilon=q$. 
  %A state $q$ in a FA $D=(\Sigma, Q, q_0, F, \delta)$  is co-deterministic for a word $w$ if  $|\{p\in Q\mid p\cdot w=q\}|=1$.
  
A symmetric use of the dot notation leads  to the following definition.  Let $w\cdot q=\{q'\mid q\in q'\cdot w\}$.  We  extend the dot notation to any set of states $S$ by $S\cdot w=\bigcup_{s\in S}s\cdot w$ and $w\cdot S =\bigcup_{s\in S} w\cdot s$. % is defined only for the states of $S$ that are co-deterministic for the word  $w$. %A same extension can be done for subsets of couples of states.
 A word $w\in \Sigma ^*$ labels a successful path in a FA $A$ if $I\cdot w\cap F\neq \emptyset$. 

%A nondeterministic finite automaton (NFA) is a $5$-tuple $B=(\Sigma,Q,q_0,F,\delta)$ where $\Sigma, Q, q_0$ and $F$ are defined as above and $\delta \ :\ Q\times \Sigma \rightarrow 2^Q$ is the transition function which can be extended to any word of $\Sigma^*$ by $\delta(p,a w)=\bigcup_{q\in\delta(p,a)}\delta(q,w)$ and $\delta(q,\varepsilon)=\{q\}$. A word $w\in \Sigma ^*$ labels a successful path in a NFA $B$ if $\delta(q_0,w)\cap F\neq \emptyset$.
 In this paper, we assume that all FA  are complete which means that for all $q\in Q$, for all $a\in \Sigma$, $\#(q\cdot a)\geq 1$. A state $q$ is accessible in a FA  if there exists a word $w\in \Sigma ^*$ such that $q\in I\cdot w$.
The language recognized by a FA $A$ is the set of words labeling a successful path in $A$. 
Two automata are said to be equivalent if they recognize the same language.  

Let $D=(\Sigma,Q_D,i_D,F_D,\cdot)$ be a DFA.
Two states $q_1,q_2$ of  $D$ are equivalent if for any word $w$ of $\Sigma^*$, $q_1\cdot w\in F_D$ if and only if $q_2\cdot w\in F_D$. Such an equivalence is denoted by $q_1\sim q_2$. A DFA is  minimal if there does not exist any equivalent DFA  with less states and it is well known that for any DFA, there exists a unique minimal equivalent one \cite{HU79}. Such a minimal DFA  can be  obtained from $D$ by computing the accessible part of the automaton $D\slash \sim=(\Sigma,Q_D\slash \sim,[i_D],F_D\slash \sim,\cdot)$ where for any $q\in Q_D$, $[q]$ is the $\sim$-class of the state $q$ and for any $a\in \Sigma$, $[q]\cdot a=[q\cdot a]$. In a minimal DFA, any two distinct states are pairwise non-equivalent.

The states of a FA are often denoted with indexed symbols and arithmetic operations can be used to compute new index from given ones. Since this index allows to point to a state of the same FA, the operations are always done modulo the states number of the FA. This is recurrent in the paper and, in general, not explicitly mentioned.

 The state complexity of a regular language $L$ denoted by $\sc(L)$ is the number of states of its minimal DFA. 
 %The  state complexity of a $k$-ary operation $\otimes$ over languages $L_1,\ldots , L_k$ is a $k$-ary function returning the state complexity of the language $\otimes (L_1,\ldots , L_k)$.
  Let ${\cal L}_n$ be the set of languages of state complexity $n$. The state complexity of a unary operation $\otimes$ is the function $\sc_{\otimes}$ associating with an integer $n$ the maximum of the state complexities of $(\otimes L)$ for $L\in {\cal L}_n$.
  A language $L\in {\cal L}_n$ is a witness (for $\otimes$) if  $\sc(\otimes L)=\sc_{\otimes}(n)$.
  %$\{L_i\}_{i\in\mathbb{N}}$ such that $\sc(\otimes L_i)=sc_{\otimes}(sc(L_i))$ for any $i\in\mathbb{N}$ is called a witness for $\otimes$. 
  This can be generalized, and the state complexity of a $k$-ary operation $\otimes$ is the $k$-ary function which associates with any tuple $(n_1,\ldots,n_k)$ the integer $\mathrm{max}\{\sc(\otimes(L_1,\ldots,L_k))|L_i\in\mathcal{L}_{n_i},\forall i\in[1,k]\}$. Then, a witness is a tuple $(L_{1},\ldots,L_{k})\in({\cal L}_{n_1}\times \cdots  \times{\cal L}_{n_k})$ such that $\sc(\otimes(L_{1},\ldots,L_{k}))=\sc_{\otimes}(n_1,\ldots,n_k)$. %for any tuple $(i_1,\cdots,i_k)\in\mathbb{N}^k$.
  An important research area consists in finding witnesses for any $(n_1,\ldots ,n_k)\in \mathbb{N}^k$.
  %modif
%  and for any combination of elementary operations. Obviously, 
 
%  \begin{clam}\label{claim1}
%   The state complexity of an operation defined as a composition of more elementary ones is upper-bounded by the composition of the corresponding elementary state complexities. 
%   \end{clam}

   For example, let us consider the ternary operation $\otimes$ defined for any three languages $L_1, L_2, L_3$ by $\otimes(L_1,L_2,L_3)=L_1\cdot(L_2\cdot L_3)$ and let $h$ be its state complexity. 
  Let $f$ be the  state complexity of $\cdot$. For any three integers $n_1,n_2,n_3$, it holds $h(n_1,n_2,n_3)\leq f(f(n_1,n_2),n_3))$ \cite{GY09}. In fact, applying the catenation on a witness does not produce a good candidate for a witness. %Moreover, following~\cite{CGKY11}, if $\circ=\cap$ then $h(n_1,n_2,n_3)=f(n_1,g(n_2,n_3))$ whereas, $h(n_1,n_2,n_3)<f(n_1,g(n_2,n_3))$ when $\circ=\cup$.

% Let $\otimes$ be an operation and $\mathrm{f}$ be its state complexity.
%
%If $\otimes$ is unary, a $\otimes$-witness for an integer $n$ is a language $L$ satisfying the two following conditions:
%\begin{enumerate}
%  \item $\mathrm{sc}(L)=n$,
%  \item $\mathrm{sc}(\otimes(L))=\mathrm{f}(n)$.
%\end{enumerate} 
%
%This can be generalized for a $k$-ary operation.
%If $\otimes$ is $k$-ary, a $\otimes$-witness for a $k$-uple $(n_1,\ldots,n_k)$ is a $k$-uple of languages $(L_1,\ldots,L_k)$ satisfying the two following conditions:
%\begin{enumerate}
%  \item $\mathrm{sc}(L_i)=n_i$ for any integer $i\leq k$,
%  \item $\mathrm{sc}(\otimes(L_1,\ldots,L_k))=\mathrm{f}(n_1,\ldots,n_k)$.
%\end{enumerate} 

%Let us denote by $\mathbb{W}$ the set containing all the $\otimes$-witnesses for any $k$-tuple.
%A $\otimes$-witness is a subset of $\mathbb{W}$ that contains at least one $\otimes$-witness for any $k$-tuple of integer.

%  The state complexity of the $k$-ary  operation   can be obtained following these two steps. Firstly
% producing an upper bound by combining the bounds of the operations composing the $k$-ary operation $\otimes$, secondly, producing a language as a set of automata which reaches this bound. This language is known as a \textit{witness} for the  composed operation.
 
In \cite{Brz13}, Brzozowski defines  a family of languages  that turns to be universal witnesses for several operations. The automata denoting these languages are called \textit{Brzozowski automata}.
We  need some background to define these automata. We 
%lud suppr will 
follow the terminology of \cite{GM08}. Let $Q=\{0,\ldots, n-1\}$ be a set. A \textit{transformation} of the set $Q$ is a mapping of $Q$ into itself. If $t$ is a transformation and $i$ an element of $Q$, we denote by $it$ the image of $i$ under $t$. A transformation of $Q$ can be represented by $t=[i_0, i_1, \ldots i_{n-1}]$ which means that $i_k=kt$ for each $0\leq k\leq n-1$ and $i_k\in Q$. A \textit{permutation} is a bijective transformation on $Q$. The \textit{identity} permutation of $Q$ is denoted by $\mathds{1}$. A \textit{cycle} of length $\ell\leq n$  is a permutation $c$, denoted   by $(i_0,i_1,\ldots i_{\ell-1})$, on a subset $I=\{i_0,\ldots ,i_{\ell-1}\}$ of $Q$  where  $i_kc=i_{k+1}$ for $0\leq k<\ell-1$ and $i_{\ell-1}c=i_0$.  A \emph{$k$-rotation}  is obtained by composing $k$ times the same cycle. In other word, we construct a $k$-rotation $r_k$ from the cycle $(i_0,\dots,i_{\ell-1})$ by setting  $i_{j}r_k=i_{j+k \mod \ell}$ for $0\leq j\leq \ell-1$. 
% A \emph{grouping} of two elements $i_j,i_k\in I$ is a $((k-j)\mod \ell)$-rotation   letting $i_k$ unchanged and obtained by iterating $(k-j)\mod \ell$ times the cycle $(i_0,\dots,i_{k-1},i_{k+1},\dots,i_{\ell-1})$. Such a grouping sends $i_j$ to $i_{k+1}$ and $i_k$ to $i_k$.
A \textit{transposition} $t=(i,j)$ is a permutation on $Q$ where $it=j$ and $jt=i$ and for every  elements $k\in Q\setminus \{i,j\}$, $kt=k$.  A \textit{contraction}  $t=\left(\begin{array}{r}i\\j\end{array}\right)$ is a transformation where  $it=j$ and  for every  elements $k\in Q\setminus \{i\}$, $kt=k$.
Then,  a Brzozowski automaton is a complete  DFA $(\Sigma, Q=\{0,\ldots , n-1\}, 0, F=\{q_f=n-1\}, \cdot)$, where any letter of $\Sigma$ induces one of the transformation among transposition, cycle over $Q$, contraction and identity.
  Let $a,b,c,d$ be distinct  symbols of $\Sigma$. As an example of Brzozowski automata (see Figure \ref{Brzo}), let \label{Brzo-def} $W_n(a,b,c,d)=(\Sigma,Q_n,0,\{q_f\},\cdot)$ where $Q_n=\{0,1,\ldots ,n-1\}$, the symbol $a$ acts as the cycle  $(0,1,\ldots, n-1)$,  $b$ acts as the transposition $(0, 1)$,  $c$ acts as the contraction $\left(\begin{array}{r}1\\0\end{array}\right)$ and   $d$ acts as $\mathds{1}$.%$1_{Q_n}$. 

 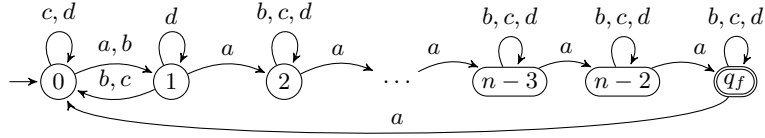
\begin{figure}[htb]
	\centerline{
		\begin{tikzpicture}[node distance=1.5cm, bend angle=25]
			\node[state,initial] (0)  {$0$};
			\node[state] (1) [right of=0] {$1$};
			\node[state] (2) [right of=1] {$2$};
			\node (etc1) [right of=2] {$\ldots$};
			\node[state, rounded rectangle] (m-3) [right of=etc1] {$n-3$};
			\node[state, rounded rectangle] (m-2) [right of=m-3] {$n-2$};
			\node[state,accepting, rounded rectangle] (m-1) [right of=m-2] {$q_f$};
			\path[->]
        (0) edge[bend left] node {$a,b$} (1)
        (1) edge[bend left] node {$a$} (2)
        (2) edge[bend left] node {$a$} (etc1)
        (etc1) edge[bend left] node {$a$} (m-3)
        (m-3) edge[bend left] node {$a$} (m-2)
        (m-2) edge[bend left] node {$a$} (m-1)
        (m-1) edge[out=-115, in=-65, looseness=.2] node[above] {$a$} (0)
		    (0) edge[out=115,in=65,loop] node {$c, d$} (0)
		    (1) edge[out=115,in=65,loop] node {$d$} (1)
		    (2) edge[out=115,in=65,loop] node {$b, c, d$} (2)
		    (m-3) edge[out=115,in=65,loop] node {$b, c, d$} (m-3)
		    (m-2) edge[out=115,in=65,loop] node {$b, c, d$} (m-2)
		    (m-1) edge[out=115,in=65,loop] node {$b, c, d$} (m-1)
        (1) edge[bend left] node[above] {$b,c$} (0)
       % (m-1) edge[bend left] node[above] {$b$} (m-2)
			;
\end{tikzpicture}
}
\caption{The automaton $W_n(a,b,c,d)$}\label{Brzo}
\end{figure}	

% !TEX root =  main.tex
\section{A bound for the state complexity of the multiple catenation}

%\subsection{Algorithm for the multiple catenation}
We first define a construction for the multiple catenation. We then compute  an upper bound for the number of states of the resulting automaton. 

\begin{definition}\label{subsetConstruction}
(Subset construction) Let $A=(\Sigma,Q_A,I_A,F_A,\cdot_A)$ be an NFA. The "subset construction" consists to build the following DFA: $B=(\Sigma,Q,I,F,\cdot)$ where
\begin{itemize}
\item $Q=2^{Q_A}$
\item $I=\{I_A\}$
\item $F=\{X\in Q|X\cap F_A\neq\emptyset\}$
\item $\forall (X,a)\in Q\times \Sigma$, $X\cdot a=\displaystyle{\bigcup_{p\in X}(p\cdot_A a)}$%\{q|\exists p\in X \mbox{ tel que } (p,a,q)\in\cdot_A\}\}$
\end{itemize}
\end{definition}

\begin{definition}\label{autoConcat}
Let $A=(\Sigma,Q_A,I_A,F_A,\cdot_A)$ and $B=(\Sigma,Q_B,I_B,F_B,\cdot_B)$ be two NFAs. We compute the  NFA $A\cdot B=(\Sigma,Q,I,F,\cdot)$  as follows:
\begin{itemize}
\item $Q=Q_A\cup Q_B$
\item $I=\left\{\begin{array}{ll}
 I_A&\mbox{ if }I_A\cap F_A=\emptyset\\
I_A\cup I_B&\mbox{ otherwise}
\end{array}\right.$
\item $F=F_B$
\item $p\cdot a=\left\{
\begin{array}{ll}
p\cdot_B a&\mbox{ if }p\in Q_B\\
p\cdot_A a&\mbox{ if }p\cdot _A a \cap F_A=\emptyset \wedge p\in Q_A\\
p\cdot_A a\cup I_B&\mbox{ otherwise} 
\end{array}
\right.$
\end{itemize}
\end{definition}

%\begin{remark}
%Dans l'AFN $A\cdot B$, tout état $q$ de $Q_B$ n'est accessible depuis $I$ qu'en passant par $i_B$,
%\begin{itemize}
%\item Soit parce que $i_B\in I$. Il existe donc $v\in L(B)$ tel que $(i_B,v,q)\in\cdot_B$.
%\item Soit depuis une transition de $\cdot_{A,B}$. Dans ce cas, il existe $u=u'a\in L(A)$ et $v\in L(B)$ tels que
%  $(i_A(\in I),u',p)\in\cdot_A$, $(p,a,i_B)\in\cdot_{A,B}$ et $(i_B,v,q)\in\cdot_B$.
%\end{itemize}
%\end{remark}

\begin{lemma}
$L(A\cdot B)=L(A)\cdot L(B)$
\end{lemma}
\begin{proof}
\begin{itemize}
\item Let $(u,v)\in L(A)\times L(B)$. If $u=\varepsilon$, $I_A\cap F_A\neq\emptyset$ and then $I_B\subset I$. We have $I_B\cdot_B v\cap F_B\neq \emptyset$, hence $I\cdot v\cap F\neq \emptyset$.
  We deduce $uv=v\in L(A\cdot B)$.
If $u=u'a$ then there exists $p\in I_A\cdot_A u'$ such that $p\cdot_A a\cap F_A\neq \emptyset$. So we have  $I_B\subset p\cdot a$. 		
Moreover,  $I_B\cdot_B v\cap F_B\neq \emptyset$. 
%Autrement dit,
  %$(i_A,u',p)\in\cdot_A\subset\cdot$, $(p,a,f_A)\in\cdot_A$ (donc, $(p,a,i_B)\in\cdot$) et
  %$(i_B,v,f_B)\in\cdot_B\subset\cdot$. 
  It follows $I_A\cdot u'\cdot a\cdot v\cap F\neq \emptyset$ and then $uv\in L(A\cdot B)$.

  \item Let $w\in L(A\cdot B)$. From the definition of $I$, we deduce that $I_B\subset I$ if and only if  $\varepsilon\in L(A)$. 
  If $I_B\subset I$ and $I_B\cdot w\cap F\neq \emptyset$   then, since the algorithm does not add any  transition from  automaton $B$ to  automaton $A$, we have $I_B\cdot_B w\cap F_B\neq \emptyset$.
   So, we have $w=uv$ with
  $u=\varepsilon\in L(A)$ and $v=w\in L(B)$.
  
  If $I_B\cap I=\emptyset$, we have $I_A\cdot w \cap  F\neq \emptyset$. From the definition of $\cdot$, we deduce that there
  exists $u',v\in \Sigma ^*$, $a\in \Sigma$, $p\in Q_A$ such that 
  $p\in I_A\cdot_A u'$ and $p\cdot_A a\cap F_A\neq \emptyset$ hence  $u=u'a\in L(A)$. Furthermore, $I_B\subset p\cdot a$ and $I_B\cdot _B v\cap F_B\neq \emptyset$ imply  $v\in L(B)$.
\end{itemize}
\end{proof} 
%On considère, par convention, l'opération $\cdot$ comme étant associative à gauche, i.e. on notera $A\cdot B\cdot C$ pour $(A\cdot B)\cdot C$.

Let us consider a sequence of complete DFAs $A_1=(\Sigma,Q_1,\{i_1\},F_1,\cdot_1)$, \ldots , $A_\alpha=(\Sigma,Q_\alpha,\{i_\alpha\},F_\alpha,\cdot_\alpha)$ we want to concatenate. We denote $A_{\overline{\alpha}}=(\Sigma,Q_{\overline{\alpha}},I_{\overline{\alpha}},F_{\overline{\alpha}},\cdot_{\overline{\alpha}})$ the NFA defined by $(\cdots((A_1\cdot A_2)\cdot A_3)\cdots)\cdot A_\alpha$. From previous lemma, we know that $L(A_1)\cdot L(A_2)\cdots L(A_\alpha)=L(A_{\overline{\alpha}})$. 

%\begin{lemma}\label{initiaux}
%Pour tout $0<j<\alpha$, ($i_{j+1}\in I_{\overline{\alpha}}\Rightarrow i_{j}\in I_{\overline{\alpha}}$).
%\end{lemma}
%
%\begin{proof}
%\end{proof}
%
%\begin{lemma}\label{initial}
%Pour tout $0<j<\alpha$ et tout $0<\delta\leq\alpha-j$, ($i_{j}\in I_{\overline{j}}\Leftrightarrow i_{j}\in I_{\overline{\alpha}}$).
%\end{lemma}
%
%\begin{proof}
%\end{proof}

Each state of the DFA $A_{\widehat{\alpha}}=(\Sigma,Q_{\widehat{\alpha}},I_{\widehat{\alpha}},F_{\widehat{\alpha}},\cdot_{\widehat{\alpha}})$ obtained by applying the subset construction to the NFA  $A_{\overline{\alpha}}$ can be partitioned and seen as a sequence of the form $(S_1,S_2,\ldots,S_{\alpha})$ where each $S_j$ is a subset of $Q_j$. It is easy to see that if such a sequence corresponds to an accessible state of $A_{\widehat{\alpha}}$, then it verifies the three following properties:
\begin{itemize}[leftmargin=1cm]
\item[\textit{P1}.] $S_1$ is a singleton (often assimilated to its unique element),
\item[\textit{P2}.] $\forall\ 0<k<\alpha,\ (S_k=\emptyset\Rightarrow S_{k+1}=\emptyset)$ (to access a DFA, we must go through its 
  predecessors),
\item[\textit{P3}.] $\forall\ 0<k<\alpha,\ (S_k\cap F_k\neq \emptyset\Rightarrow i_{k+1}\in S_{k+1})$ (because of the
  transitions built in Definition \ref{autoConcat}).
\end{itemize}

A sequence which verifies properties \textit{P1}, \textit{P2} and \textit{P3} is called a \textit{valid} sequence. A state associated to such a sequence is called a \textit{valid} state. We now evaluate an upper bound for the number of valid states, that is an upper bound for the state complexity of multiple catenation.

\subsection{Counting states}
%Following the previous sections, the state complexity for the catenation of
%$k$ automata of size $n_{0},n_{1},\dots,n_{k-1}$ is equal to the number of $k$-tuples
%$(i,S_{1},\dots,S_{k-1})$ satisfying
%\begin{enumerate}
%\item $i\in [n_{0}]$
%\item each $S_{j}$ is a subset of $[n_{j}]$
%\item $S_{j}=\emptyset\Rightarrow S_{j+1}=\emptyset$ for each $j\in\{1,\dots,k-2\}$,
%\item $i=n_{0}-1\Rightarrow 0\in S_{1}$,
%\item $n_{j}-1\in S_{j}\Rightarrow 0\in S_{j+1}$, for each  $j\in\{1,\dots,k-2\}$,
%\end{enumerate} 
%where $[n]$ denotes thet set $\{0,\dots,n-1\}$.\\
Let us notice that the  maximum  number of valid states is reached when each DFA has only one final state. So the number of valid states is upper bounded by $\#(\mathcal T_{\alpha})$
where $\mathcal T_{\alpha}$ is the set of valid sequences $(S_1,\ldots , S_\alpha)$ denoting the states  of the catenation of
$\alpha$ such automata of size $n_{1},\dots,n_{\alpha}$.% and with respective final states $n_1-1,\ldots , n_\alpha-1$.

%In the  aim to compute its cardinal, we splits the set $\mathcal T_{\alpha}$ of the states into several disjoint sets. We proceed as follows.
%First we define  the following sets:
%\begin{equation}\mathcal T^{+}_{1}=\{(n_{1})\},\end{equation} 
%\begin{equation}\mathcal T^{-}_{1}=\{(i):0\leq i<n_{1}-1\},\end{equation}
%\begin{equation}\mathcal T^{+}_{j}=\{(i,S_{2},\dots,S_{j}):
%i\in[1,n_{1}], S_{\ell}\subset [1,n_{\ell}], S_{\ell}\neq \emptyset, n_{i}-1\in S_{j}\},\end{equation}
%\begin{equation}
%\mathcal S^{-}_{j}=\{(i,S_{1},\dots,S_{j}):i\in[n_{0}], S_{\ell}\subset [n_{\ell}], n_{i}-1\not\in S_{j}\},
%\end{equation}
%for $1\leq j\leq k-1$.\\

For any $j\in \mathbb{N}\setminus \{0\}$, let $\mathcal T^{+}_{j}$ (resp $\mathcal T^{-}_{j}$) be the subset of   $\mathcal T_{j}$ constituted with the sequences of non empty sets $(S_1,\ldots, S_j)$ with $q_f\in S_j$ (resp. $q_f\not\in S_j$).
A fast examination of the elements of $\mathcal T_\alpha$ gives
\begin{lemma}\label{Sk-1}
The set $\mathcal T_\alpha$ is the disjoint union
\begin{equation}
\mathcal T_\alpha=\mathcal T^+_\alpha\uplus\biguplus_{j=1}^{\alpha}\{(S_1,\dots,S_{j},\emptyset,\dots,\emptyset):(S_{1},\dots,S_{j})\in \mathcal T^{-}_{j}\}.
\end{equation}
\end{lemma}
Notice that for each $1<j\leq \alpha$, $\mathcal T^{-}_{j}$ splits into two disjoint sets
\begin{equation}\label{T-}
\mathcal T^{-}_{j}=\mathcal T^{--}_{j}\uplus \mathcal T^{+-}_{j}
\end{equation}
where
\begin{equation}\label{T--}
\mathcal T^{--}_{j}=\{(S_{1},\dots,S_{j})\in 
\mathcal T^{-}_{j}\mid  (S_{1},\dots,S_{j-1})\in \mathcal T^-_{j-1} \},
\end{equation}
and
\begin{equation}\label{T+-}
\mathcal T^{+-}_{j}=\{(S_{1},\dots,S_{j})\in \mathcal T^-_{j}\mid (S_{1},\dots,S_{j-1})\in 
\mathcal T^{+}_{j-1}\}.
\end{equation}
Also $\mathcal T^{+}_{j}$ splits into two disjoint sets
\begin{equation}%\label{T+}
\mathcal T^{+}_{j}=\mathcal T^{-+}_{j}\uplus \mathcal T^{++}_{j}
\end{equation}
where
\begin{equation}\label{T-+}
\mathcal T^{-+}_{j}=\{(S_{1},\dots,S_{j})\in 
\mathcal T^{+}_{j}\mid  (S_{1},\dots,S_{j-1})\in \mathcal T^-_{j-1} \},
\end{equation}
and
\begin{equation}\label{T++}
\mathcal T^{++}_{j}=\{(S_{1},\dots,S_{j})\in \mathcal T^+_{j}\mid (S_{1},\dots,S_{j-1})\in 
\mathcal T^{+}_{j-1}\}.
\end{equation}
%Denoting by $s_{j}$ (resp. $s^{\pm}_{j}$, $\tilde s^{\pm}_{j}, \hat s^{\pm}_{j}$) the  cardinal of $\mathcal S_{j}$ (resp. $\mathcal S^{\pm}_{j}$, $\widetilde {\mathcal S^{\pm}_{j}}$, 
%$\widehat {\mathcal S^{\pm}_{j}}$), we obtain
\begin{proposition}\label{Pcr}
	We have the following crossing recurrence 
\begin{equation}\label{crossingrec}
\left\{\begin{array}{l}
\#\mathcal T^{-}_{1}=n_{1}-1,\\ 
\#\mathcal T^{+}_{1}=1,\\
\#\mathcal T^{-}_{j}=(2^{n_{j}-1}-1)(\#\mathcal T^{-}_{j-1})+2^{n_{j}-2}(\#\mathcal T^{+}_{j-1})\\
\#\mathcal T^{+}_{j}=2^{n_{j}-1}(\#\mathcal T^{-}_{j-1})+2^{n_{j}-2}(\#\mathcal T^{+}_{j-1})=\#\mathcal T^{-}_{j}+\#\mathcal T^{-}_{j-1}.
\end{array}
\right.
\end{equation}
Furthermore, we have
\begin{equation}\label{sk-1}
	\#\mathcal T_{\alpha}=\sum_{j=1}^{\alpha}\#\mathcal T^{-}_{j}+\#\mathcal T^{+}_{j}.
\end{equation}
\end{proposition}
\begin{proof}
	It suffices to remark that formulas (\ref{T--}), (\ref{T+-}), (\ref{T-+}), (\ref{T++}) imply 
\begin{equation}{}
\#\mathcal T^{--}_{j}=(2^{n_{j}-1}-1)(\#\mathcal T^{-}_{j-1}),\ \ \ \#\mathcal T^{+-}_{j}=2^{n_{j}-2}(\#\mathcal T^{+}_{j-1}),
\end{equation}
{}
\begin{equation}
\#\mathcal T^{++}_{j}=2^{n_{j}-1}(\#\mathcal T^{-}_{j-1}),\ \ \ \#\mathcal T^{-+}_{j}=2^{n_{j}-2}(\#\mathcal T^{+}_{j-1}).
\end{equation}
Formula (\ref{sk-1}) comes immediately from Lemma \ref{Sk-1}. 
\end{proof}

\begin{example}\label{exampleJG}
	Applying proposition \ref{Pcr}, we find after simplification
	\[{}
	\#\mathcal T_{2}={2}^{{n_2}-1} \left( 2\,{ n_1}-1 \right) 
	\]
	\[{}
	\#\mathcal T_{3}={  n_1}-1+\frac38\,{2}^{{  n_3}+{  n_2}}(2n_{1}-1)+{2}^{{  n_3}}(1-  n_1).
	\]
	\[{}\begin{array}{rcl}
	\#\mathcal T_{4}&=&{2}^{{  n_4}}{  n_1}+{\frac {9}{16}}\,{2}^{{  n_4}+{  n_3}+{  n_2
}}{  n_1}- \frac34\,{2}^{{  n_4}+{  n_3}}{  n_1}+{2}^{{  n_2}-1}{  
n_1}-\frac14\,{2}^{{  n_2}}-{2}^{{  n_4}}-{\frac {9}{32}}\,{2}^{{  n_4}+
{  n_3}+{  n_2}}+\frac14\,{2}^{{  n_4}+{  n_1}}\\&&+\frac34\,{2}^{{  n_4}+{
  n_3}}-{2}^{{  n_4}-1+{  n_2}}{  n_1}.\end{array}
	\]
	\[{}
	\begin{array}{rcl}
	\#\mathcal T_{5}&=&-1+{ n_1}-\frac3{16}\,{2}^{{ n_3}+{ n_2}}+\frac38\,{2}^{{ n_3}+{ n_2}}
{ n_1}+{2}^{{ n_3}-1}-{2}^{{ n_5}-1+{ n_3}}+\frac3{16}\,{2}^{{ n_5
}+{ n_4}+{ n_2}}-{\frac {27}{128}}\,{2}^{{ n_5}+{ n_4}+{ n_3
}+{ n_2}}\\
&&+{\frac {9}{16}}\,{2}^{{ n_5}+{ n_4}+{ n_3}}+{2}^{{
 n_5}}-\frac34\,{2}^{{ n_5}+{ n_4}}+\frac3{16}\,{2}^{{ n_5}+{ n_3}+{
 n_2}}-\frac14\,{2}^{{ n_5}+{ n_2}}-{2}^{{ n_3}-1}{ n_1}+{2}^{{
 n_5}-1+{ n_2}}{ n_1}+{2}^{{ n_5}-1+{ n_3}}{ n_1}\\
 &&-\frac38\,{2} 
^{{ n_5}+{ n_4}+{ n_2}}{ n_1}+{\frac {27}{64}}\,{2}^{{ n_5}+
{ n_4}+{ n_3}+{ n_2}}{ n_1}-{\frac {9}{16}}\,{2}^{{ n_5}+{
 n_4}+{ n_3}}{ n_1}-{2}^{{ n_5}}{ n_1}-\frac38\,{2}^{{ n_5}+{
 n_3}+{ n_2}}{ n_1}+\frac34\,{2}^{{ n_5}+{ n_4}}{ n_1}
	\end{array}
	\]
\end{example}
\subsection{Expanded formula}
As one can see, the first values of  $\#{\cal T}_n$ can be recursively computed but seem to be tedious. Some regularities can be observed which allow us to propose a combinatorial description of $\#{\cal T}_n$.% are very complex to explain. In this section, we give the tools in order to obtain expanded formula of $\#{\cal T}_n$ that can be directly computed.

\noindent
For the sake of simplicity we consider the formal variables $x_j$ for $j>0$, $y$ and $z$ and define the multivariable polynomials $s^-_j$, $s^+_j$ and $s_j$ for $j\geq 0$.
\def\st{\mathtt s}\def\rt{\mathtt r}

\centerline{$\begin{array}{ll}
	\st^{-}_{0}=z,\st^{+}_{0}=2y,\\
	\st^{-}_{j}=(x_{j}-1)\st^{-}_{j-1}+\frac12x_{j}\st^{+}_{j-1},&\mbox{ for }j>0 \mbox{ and }\\
	\st^{+}_{j}=x_{j}\st^{-}_{j-1}+\frac 12x_{j}\st_{j-1}^{+}=\st_{j}^{-}+\st_{j-1}^{-},&\mbox{ for }j>0.
\end{array}$}

We set also 
\begin{equation}\label{defst}
\st_{\alpha-1}=\st^{+}_{\alpha-1} + \sum_{j=0}^{\alpha-1}\st^{-}_{j}.
\end{equation}
Notice that we recover $\#\mathcal T_{\alpha}$ from $\st_{\alpha-1}$ by setting $x_{j}=2^{n_{j+1}-1}$ 
for $1\leq j\leq \alpha-1$, $y=\frac12$, and $z=n_{1}-1$.
For technical reasons, we will use the polynomial
\begin{equation}\label{defrt}
	\rt_{\alpha-1}=\st_{\alpha-1}|_{x_{\alpha-1}\rightarrow \frac12x_{\alpha-1}}.
\end{equation}
where $P|_{x\rightarrow t}$ means that each occurrence of the variable $x$ is replaced by the expression $t$ in the polynomial $P$.

\begin{example}
~
	\begin{enumerate}
	\item $\rt_{0}=2y+z$
	\item $\rt_{1}={\it x_1}\,y+{\it x_1}\,z$
	\item $\rt_{2}=z+\frac32 \,{\it x_2}\,{\it x_1}\,y+\frac32 \,{\it x_2}\,{\it x_1}\,z-{\it x_2}\,z${}
	\item $\rt_{3}=\frac94 \,{\it x_3}\,{\it x_2}\,{\it x_1}\,y+\frac94 \,{\it x_3}\,{\it x_2}\,{\it x_1}
\,z-{\it x_3}\,{\it x_1}\,y+{\it x_1}\,y+{\it x_1}\,z-\frac32 \,{\it x_3}\,{\it 
x_2}\,z+{\it x_3}\,z-{\it x_3}\,{\it x_1}\,z
${}
\item $\rt_{4}=-\frac94 \,{\it x_4}\,{\it x_3}\,{\it x_2}\,z-\frac32 \,{\it x_4}\,{\it x_3}\,{\it x_1
}\,y-\frac32 \,{\it x_4}\,{\it x_3}\,{\it x_1}\,z-\frac32 \,{\it x_4}\,{\it x_2}\,{
\it x_1}\,y-\frac32 \,{\it x_4}\,{\it x_2}\,{\it x_1}\,z+z+{\frac {27}{8}}\,{
\it x_4}\,{\it x_3}\,{\it x_2}\,{\it x_1}\,y+{\frac {27}{8}}\,{\it x_4}\,{
\it x_3}\,{\it x_2}\,{\it x_1}\,z+\frac32 \,{\it x_4}\,{\it x_3}\,z+{\it x_4}\,{
\it x_2}\,z+{\it x_4}\,{\it x_1}\,y+{\it x_4}\,{\it x_1}\,z+\frac32 \,{\it x_2}\,
{\it x_1}\,z+\frac32 \,{\it x_2}\,{\it x_1}\,y-{\it x_2}\,z-{\it x_4}\,z
$	
\end{enumerate}
\end{example}

%\begin{lemma}\label{Lrt2st}
It is easy to check that 
	\begin{equation}\label{Lrt2st}
		\rt_{\alpha-1}=\sum_{j=0}^{\alpha-1}\st^{-}_{j}.
	\end{equation}
%\end{lemma}
%\begin{proof}
%	First observe that
%	\begin{eqnarray}
%	\st^{-}_{\alpha-1}|_{x_{\alpha-1}\rightarrow \frac12 x_{\alpha-1}}&=&
%	\frac12y\sum_{c\vdash \alpha-1\atop c_{|c|}=1}(-1)^{|c|+\alpha-1}\left(3\over 2\right)^{\Theta(c)}[c]+
%	y\sum_{c\vdash \alpha-1\atop c_{|c|}\neq 1}(-1)^{|c|+\alpha-1}\left(3\over 2\right)^{\Theta(c)}[c]\\
%	&&+\frac12z\sum_{c\vdash \alpha\atop c_{|c|}=1}(-1)^{|c|+\alpha}
%	\left(3\over 2\right)^{\tilde\Theta(c)}\{c\}
%	+z\sum_{c\vdash \alpha\atop c_{|c|}\neq1}(-1)^{|c|+\alpha}
%	\left(3\over 2\right)^{\tilde\Theta(c)}\{c\}.
%	\end{eqnarray}
%	and
%	\begin{equation}
%		\st^{+}_{\alpha-1}|_{x_{\alpha-1}\rightarrow \frac12 x_{\alpha-1}}=\frac12\st^{+}_{\alpha-1}=
%		\frac12y\sum_{c\vdash \alpha-1\atop c_{|c|}=1}(-1)^{|c|+\alpha-1}\left(3\over 2\right)^{\Theta(c)}[c]+
%		\frac12z\sum_{c\vdash k\atop c_{|c|}=1}(-1)^{|c|+\alpha}
%	\left(3\over 2\right)^{\tilde\Theta(c)}\{c\}.
%	\end{equation}
%So 
%\begin{equation}
%	\left(\st^{-}_{\alpha-1}+\st^{+}_{\alpha-1}\right)|_{x_{\alpha-1}\rightarrow \frac12 x_{\alpha-1}}=\st^{-}_{\alpha-1}.
%\end{equation}
%We deduce
%\begin{equation}
%	\rt_{\alpha-1}=\sum_{i=0}^{\alpha-2}\st^{-}_{i}+\left(\st^{-}_{\alpha-1}+\st^{+}_{\alpha-1}\right)|_{x_{\alpha-1}
%	\rightarrow \frac12 x_{\alpha-1}}=
%	\sum_{i=0}^{\alpha-1}\st^{-}_{i}.
%\end{equation}
%\end{proof}

Let us recall some notation. A composition is a finite list of  positive integers $c=(c_{1},\dots,c_{\ell})$. When there is no ambiguity, we  denote $c$ by $c_1\cdots c_\ell$. The length of the composition is $|c|=\ell$. We denote by $c\vDash n$ if and only if $c_{1}+\dots+c_{\ell}=n$.
We define also  $\Theta(c)=\#\{j\mid c_{j}=1, 1\leq j\leq \ell-1\}$, %and 
	 $[c]=x_{1}x_{1+c_{1}}x_{1+c_{1}+c_{2}}\cdots x_{1+c_{1}+\cdots+c_{\ell-1}}$ for $\ell >0$ and $[()]=0$.
	 	 
%For simplicity, we denote a composition by the word of its entries.
\begin{example}
	Consider $c=212113$, one has $c\vDash 10$, $|c|=6$, $\Theta(c)=3$, $[c]=x_{1}x_{3}x_{4}x_{6}x_{7}x_{8}$.
\end{example}
We define the polynomials
\begin{equation}
	\mathtt{m}_{i}=\sum_{c\vDash i}(-1)^{|c|+i}\left(\frac32\right)^{\Theta(c)}[c].
\end{equation}
\begin{lemma}\label{lmy}
	For $i>1$ we have 
	\[{}\left\{\begin{array}{l}
	\mathtt{m}_{0}=0\\
	\mathtt{m}_{1}=x_1\\
	\mathtt{m}_{i}=\left(\frac32x_{i}-1\right)\mathtt{m}_{i-1}+\frac12 x_{i}\mathtt{m}_{i-2}.
	\end{array}\right.
	\]
\end{lemma}
\begin{proof}
Let $c\vDash i$. If $c=c'1$ then $c'\vDash i-1$, $|c'|=|c|-1$, $[c]=[c']x_{i}$ and
\begin{equation}
	\Theta(c)=\left\{\begin{array}{ll}\Theta(c')+1&\mbox{ if }c_{|c'|}= 1\\ 
	\Theta(c')&\mbox{ if }c_{|c'|}\neq 1\end{array}\right.
\end{equation}
If $c=c'p$ with $p>1$ then setting $c''=c'(p-1)$, we have  $c''\vDash i-1$ with $|c''|=|c|$, $[c]=[c'']$ and $\Theta(c'')=\Theta(c)$.{}
According to the previous remarks, the sum splits as follows
\begin{equation}
\begin{array}{rcl}
	\mathtt m_{i}&=&\displaystyle \sum_{c\vDash i\atop c_{|c|}=c_{|c|-1}=1}(-1)^{|c|+i}\left(3\over 2\right)^{\Theta(c)}[c]+
	\sum_{c\vDash i\atop c_{|c|}=1, c_{|c|-1}>1}(-1)^{|c|+i}\left(3\over 2\right)^{\Theta(c)}[c]+
	\displaystyle\sum_{c\vDash i\atop c_{|c|}\neq 1}(-1)^{|c|+i}\left(3\over 2\right)^{\Theta(c)}[c]\\&=&
	\displaystyle\sum_{c\vDash i-1\atop c_{|c|}=1}(-1)^{|c|+1+i}\left(3\over 2\right)^{\Theta(c)+1}[c]x_{i}+
	\sum_{c\vDash i-1\atop c_{|c|}>1}(-1)^{|c|+1+i}\left(3\over 2\right)^{\Theta(c)}[c]x_{i}+
	\displaystyle\sum_{c\vDash i-1}(-1)^{|c|+i}\left(3\over 2\right)^{\Theta(c)}[c]\\
	&=& \displaystyle\left(\frac32x_{i}-1\right)\sum_{c\vDash i-1\atop c_{|c|}=1}(-1)^{|c|+i+1}\left(3\over 2\right)^{\Theta(c)}[c]
	+(x_{i}-1)\sum_{c\vDash i-1\atop c_{|c|}>1}(-1)^{|c|+i+1}\left(3\over 2\right)^{\Theta(c)}[c]\\
	&=& \displaystyle\left(\frac32x_{i}-1\right)\sum_{c\vDash i-1}(-1)^{|c|+i+1}\left(3\over 2\right)^{\Theta(c)}[c]
	+\frac12x_{i}\sum_{c\vDash i-1\atop c_{|c|}>1}(-1)^{|c|+i}\left(3\over 2\right)^{\Theta(c)}[c]\\
\end{array}
\end{equation}
And similarly, we also have
\begin{equation}
\sum_{c\vDash i-1\atop c_{|c|}>1}(-1)^{|c|+i}\left(3\over 2\right)^{\Theta(c)}[c]=
\sum_{c\vDash i-2}(-1)^{|c|+i}\left(3\over 2\right)^{\Theta(c)}[c].
\end{equation}
We deduce
\begin{equation}\label{eq2my}
	\mathtt m_{i}=(\frac32 x_{i}-1)\mathtt m_{i-1}+\frac12x_{i}\mathtt m_{i-2}.
\end{equation}
%$\Box$
\end{proof}
\begin{lemma}\label{sy-}
	Let $\st^{y}_{i}$ be the coefficient of $y$ in $\st^{-}_{i}$. We have
	\[{}
	\st^{y}_{i}=\sum_{c\vDash i}(-1)^{|c|+i}\left(\frac32\right)^{\Theta(c)}[c]
	\]
\end{lemma}
\begin{proof}
	It suffices to remark that $\st^{y}_{i}$ satisfies the same recurrence and initial conditions as $\mathtt m_{i}.$
	%$\Box$
\end{proof}
The following result is obtained in a very similar way.
\begin{lemma}\label{sz-}
	Let $\st^{z}_{i}$ be the coefficient of $z$ in $\st^{-}_{i}$. We have
	\[{}
	\st^{z}_{i}=\sum_{c\vDash i+1}(-1)^{|c|+i+1}\left(\frac32\right)^{\tilde\Theta(c)}\{c\}
	\]
	where $\{c\}=x_{c_{1}}x_{c_{1}+c_{2}}\cdots x_{c_{1}+\cdots+c_{|c|-1}}$ for any non-empty list $c$ and $\tilde\Theta(c)=\#\{i\in\{2,\dots,|c|-1\}\mid c_{i}=1\}$.
\end{lemma}
\begin{proof}
The proof follows the same pattern as in Lemma \ref{lmy} and \ref{sy-}, since the sequence of the $\mathtt{s}^z_i$ admits the same recurrence as the sequence of the $\mathtt{s}^y_i$ but with the initial conditions $\mathtt{s}^z_0= 1$ and $\mathtt{s}^z_1=x_1-1$.
\end{proof}
From Lemmas \ref{sy-} and \ref{sz-}, one obtains
\begin{equation}\label{eqs-}
	\st^{-}_{i}=y\sum_{c\vDash i}(-1)^{|c|+i}\left(3\over 2\right)^{\Theta(c)}[c]+z\sum_{c\vDash i+1}(-1)^{|c|+i+1}
	\left(3\over 2\right)^{\tilde\Theta(c)}\{c\}.
\end{equation}
%and
%\begin{equation}\label{eqs+}
%\st^{+}_{i}=y\sum_{c\vDash i\atop c_{|c|}=1}(-1)^{|c|+i}\left(3\over 2\right)^{\Theta(c)}[c]+
%z\sum_{c\vDash i+1\atop c_{|c|}=1}(-1)^{|c|+i+1}\left(3\over 2\right)^{\tilde\Theta(c)}\{c\}.
%\end{equation}
\begin{example}	
The compositions of $4$ and $5$ are summarized in the following tables:\\
\centerline{\begin{minipage}[t]{0.3\linewidth}
	%The compositions of $4$ are summarized in the following table:
	\[\begin{array}{c|c|c}
		c&\Theta(x)&[c]\\
		4&0&x_{1}\\
		31&0&x_{1}x_{4}\\
		13&1&x_{1}x_{2}\\
		211&1&x_{1}x_{3}x_{4}\\
		121&1&x_{1}x_{2}x_{4}\\
		112&2&x_{1}x_{2}x_{3}\\
		22&0&x_{1}x_{3}\\
		1111&3&x_{1}x_{2}x_{3}x_{4}
	\end{array}\]
\end{minipage}
%\hfill
\begin{minipage}[t]{0.3\linewidth}
%	The compositions of $5$ are summarized in the following table :
	\[\begin{array}{c|c|c}
		c&\tilde\Theta(x)&\{c\}\\
		5&0&1\\
		41&0&x_{4}\\
		14&0&x_{1}\\
		32&0&x_{3}\\
		23&0&x_{2}\\
		311&1&x_{3}x_{4}\\
		131&0&x_{1}x_{4}\\
		113&1&x_{1}x_{2}\\
		221&0&x_{2}x_{4}\\
		212&1&x_{2}x_{3}\\
		122&0&x_{1}x_{3}\\
		2111&2&x_{2}x_{3}x_{4}\\
		1211&1&x_{1}x_{3}x_{4}\\
		1121&1&x_{1}x_{2}x_{4}\\
		1112&2&x_{1}x_{2}x_{3}\\
		11111&3&x_{1}x_{2}x_{3}x_{4}
	\end{array}\]
\end{minipage}}

So we have
\[{}\begin{array}{rcl}
\st^{-}_{4}&=&y\left(-x_{1}+x_{1}x_{4}+\frac32x_{1}x_{2}-\frac32x_{1}x_{3}x_{4}-\frac32x_{1}x_{2}x_{4}
-\frac94x_{1}x_{2}x_{3}+\frac{27}8x_{1}x_{2}x_{3}x_{4}\right){}
+\\
&&z\left(1-x_{4}-x_{1}-x_{2}-x_{3}+\frac32x_{3}x_{4}+x_{1}x_{4}+\frac32x_{1}x_{2}
\right.\\
&&\left.+x_{2}x_{4}+\frac32x_{2}x_{3}+x_{1}x_{3}-\frac94 x_{2}x_{3}x_{4}-\frac32x_{1}x_{3}x_{4}
-\frac32x_{1}x_{2}x_{4}-\frac94x_{1}x_{2}x_{3}+\frac{27}8x_{1}x_{2}x_{3}x_{4}\right).\end{array}
\]
%For the computation of $\st^{+}_{4}$, we use only the monomials having a factor equal to $x_{4}$:
%\[{}\begin{array}{rcl}
%\st^{+}_{4}&=&yx_{4}\left(x_{1}-\frac32x_{1}x_{3}-\frac32x_{1}x_{2}+
%\frac{27}8x_{1}x_{2}x_{3}\right){}
%+\\
%&&zx_{4}\left(-1+\frac32x_{3}+x_{1}+x_{2}-\frac94 x_{2}x_{3}-\frac32x_{1}x_{3}
%-\frac32x_{1}x_{2}+\frac{27}8x_{1}x_{2}x_{3}\right).\end{array}
%\]
\end{example}
\begin{theorem}
The number of valid states in the catenation of $\alpha$ automata of size $n_{1}$, $n_{2},\dots, n_{\alpha}$ 
is bounded by the number obtained by setting $x_{1}=2^{n_{2}-1},\dots, x_{\alpha-2}=2^{n_{\alpha-1}-1}$, $x_{\alpha-1}=2^{n_{\alpha}}$,
 $y=\frac12$ and 
$z=n_{1}-1$ in the polynomial
	\begin{equation}\label{devrt}\rt_{\alpha-1}=y\sum_{c=c'm\vDash \alpha-1\atop m\mbox{ }odd}(-1)^{|c|+\alpha-1}
	\left(3\over 2\right)^{\Theta(c)}[c]
	+z\sum_{c=c'm\vDash k\atop m\mbox{ }odd}(-1)^{|c|+\alpha}\left(3\over 2\right)^{\tilde \Theta(c)}\{c\}.\end{equation}
\end{theorem}
\begin{proof}
	From (\ref{defst}) and (\ref{defrt}) one has only to prove formula (\ref{devrt}). 
	Suppose that $c\vDash p$, we have $[c1]=[c2]=\dots=[c(\alpha-p-1)]$ and $\Theta(c1)=\cdots=\Theta(c(\alpha-p-1))$.{}
	Hence, by telescoping, the coefficient of the monomial $[c(\alpha-p-1)]$ in $\sum_{i=0}^{\alpha-1}\st^y_{i}$ 
	equals $0$ if $\alpha-p-1$ is even and equals $(-1)^{|c|+\alpha-1}\left(\frac32\right)^{\Theta(c(\alpha-1))}$ if $\alpha-p-1$ is odd. Hence,
	\begin{equation}
		\sum_{i=0}^{\alpha-1}\st^{y}_{i}=\sum_{c=c'm\vDash \alpha-1\atop m\mbox{ }odd}(-1)^{|c|+\alpha-1}
		\left(3\over 2\right)^{\Theta(c)}[c].
	\end{equation}
	For similar reasons, we find
	\begin{equation}
		\sum_{i=0}^{\alpha-1}\st^{z}_{i}=\sum_{c=c'm\vDash \alpha\atop m\mbox{ }odd}(-1)^{|c|+\alpha}
		\left(3\over 2\right)^{\tilde \Theta(c)}\{c\}.
	\end{equation}
	Since 
	\begin{equation}
		\sum_{i=0}^{\alpha-1}\st^{-}_{i}=y\sum_{i=0}^{\alpha-1}\st^{y}_{i}+z\sum_{i=0}^{\alpha-1}\st^{z}_{i},
	\end{equation}
	Equation \ref{Lrt2st} allows to conclude.%$\Box$
\end{proof}
\begin{example}\rm
		The following tables summarize the compositions of $3$ and $4$ such that the last entry is odd:\\
\centerline{\begin{minipage}[t]{0.3\linewidth}
	\[\begin{array}{c|c|c}
		c&\Theta(x)&[c]\\
		3&0&x_{1}\\
		21&1&x_{1}x_{3}\\
		111&2&x_{1}x_{2}x_{3}
	\end{array}\]
%The following table summarizes the compositions of $4$ such that the last entry is odd:
\end{minipage}
\begin{minipage}[t]{0.3\linewidth}
	\[\begin{array}{c|c|c}
		c&\tilde\Theta(x)&\{c\}\\
		31&0&x_{3}\\
		13&0&x_{1}\\
		211&1&x_{2}x_{3}\\
		121&0&x_{1}x_{3}\\
		1111&2&x_{1}x_{2}x_{3}
	\end{array}\]
\end{minipage}}
	
	So we obtain
	\[{}
	\rt_{3}=y(x_{1}+\frac94x_{1}x_{2}x_{3}-x_{1}x_{3})+(x_{1}+\frac94x_{1}x_{2}x_{3}-\frac32 x_{2}x_{3}-x_{1}x_{3}+x_{3})z
	\]
and then the cardinal of $\mathcal T_{4}$ is
\[{}
\frac12(2^{n_{2}-1}+\frac942^{n_{2}+n_{3}+n_{4}-2}-2^{n_{2}+n_{4}-1})
+(2^{n_{2}-1}+\frac942^{n_{2}+n_{3}+n_{4}-2}-\frac322^{n_{3}+n_{4}-1}-2^{n_{2}+n_{4}-1}+2^{n_{4}})(n_{1}-1).
\]
\end{example}

%\input{computing}
% !TEX root =  main.tex
\section{A $(\alpha+1)$-letters witness for the catenation of $\alpha$ automata}

In this section, we give  a family of witnesses automata for the multiple catenation. These witnesses are Brzozowski automata computed with the operations given in Table \ref{tableau}. Let us recall that for an automaton, $\mathds{1}$ stands for  the identity, $p$ for the cycle $(0,\ldots, n-1)$ on $Q=\{0,\ldots, n-1\}$, $t$ for the transposition $(0,1)$ of the two first states and $c$ for the contraction of state $1$ to  state $0$.

%%%%%Tableau
\begin{table}[H]
\centerline{  \begin{tikzpicture}[scale=0.5,node distance=1cm,shorten >=0pt]   
  \draw (1,1)rectangle (9,10);
%	      \foreach \x in {1,...,5} {
%	        \foreach \y in {1,...,6} {
%	          \pgfmathparse{\x+1} \let\z\pgfmathresult
%	          \pgfmathparse{\y+1} \let\t\pgfmathresult
%	          \draw[fill=gray!40] (\x,\y) rectangle (\x+1,\y+1);
%	        }
%	      }  
%	      \foreach \x/\y in {2/3, 3/1 ,4/1 ,4/4} {
%	        \pgfmathparse{\x+1} \let\z\pgfmathresult
%	        \pgfmathparse{\y+1} \let\t\pgfmathresult
%	        \draw[fill=white] (\x,\y) rectangle (\x+1,\y+1);
%	        \draw (\x,\y) -- (\z,\t);
%	        \draw (\z,\y) -- (\x,\t);
%	      }  	 
	\draw (1.5,10.5) node (q0) {$A_1$};
	\draw (2.5,10.5) node (q1) {$A_2$};
	\draw (3.5,10.5) node (q2) {$A_3$};
	\draw (6.5,10.5) node (q3) {$\ldots$};
%	\draw[dotted] (4.5,10.5) -- (7.5,10.5);
	\draw (8.5,10.5) node (q4) {$A_\alpha$};
	\draw (0.2,9.5) node (r0) {$\sigma_1$};
	\draw (0.2,8.5) node (r1) {$\sigma_2$};
	\draw (0.2,7.5) node (r2) {$\sigma_3$};
	\draw (0.2,5.5) node (r3) {$\ldots$};
	\draw (0.2,2.5) node (r4) {$\sigma_\alpha$};
	\draw (0.2,1.5) node (r5) {$\sigma_{\alpha+1}$};
	\draw (1.5,9.5) node (q0) {$t$};
	\draw (2.5,9.5) node (q1) {$c$};
	\draw (3.5,9.5) node (q2) {$\mathds{1}$};
%	\draw (4.5,9.5) node (q3) {$\mathds{1}$};
%	\draw (6.5,9.5) node (q4) {$\ldots$};
	\draw (3.8,9.5) -- (8.2,9.5);
	\draw (8.5,9.5) node (q4) {$\mathds{1}$};
	\draw (1.5,8.5) node (q0) {$p$};
	\draw (2.5,8.5) node (q1) {$t$};
	\draw (3.5,8.5) node (q2) {$c$};
%	\draw (4.5,8.5) node (q3) {$\mathds{1}$};
%	\draw (6.5,8.5) node (q4) {$\ldots$};
%	\draw (8.5,8.5) node (q4) {$\mathds{1}$};
%
	\draw (1.5,7.5) node (q0) {$\mathds{1}$};
	\draw (2.5,7.5) node (q1) {$p$};
	\draw (3.5,7.5) node (q2) {$t$};
	\draw (4.5,7.5) node (q3) {$c$};
%	\draw (6.5,7.5) node (q4) {$\ldots$};
%	\draw (8.5,7.5) node (q4) {$\mathds{1}$};
%
  \draw [dashed] (4,7) -- (7,4);
  \draw (8.5,2.5) node (q0) {$t$};
  \draw (7.5,3.5) node (q0) {$t$};
  \draw [dashed] (3,7) -- (7,3);
  \draw (8.5,1.5) node (q0) {$p$};
  \draw (7.5,2.5) node (q0) {$p$};
  \draw [dashed] (5,7) -- (7,5);
  \draw (8.5,3.5) node (q0) {$c$};
  \draw (7.5,4.5) node (q0) {$c$};
  \draw  (1.5,7.2) -- (1.5,1.8);
  \draw  (1.8,1.5) -- (7.2,1.5);
  \draw (1.8,7.2) -- (7.2,1.8);
  \draw (1.5,1.5) node (q0) {$\mathds{1}$};
  \draw (7.5,1.5) node (q0) {$\mathds{1}$};
  \draw (8.5,4.5) node (q0) {$\mathds{1}$};
  \draw  (3.8,9.2) -- (8.2,4.8);
  \draw  (8.5,9.2) -- (8.5,4.8);
  \draw (3,3.5) node (q0) {$\mathds{1}$};
  \draw (7,7.5) node (q0) {$\mathds{1}$};
\end{tikzpicture}}
\caption{Brzozowski witnesses for the multiple catenation}
\label{tableau}
\end{table}  

As a consequence, for any automaton $A_k$, $\sigma_{k-1}$ acts as the contraction $\left(\begin{array}{l}1\\0\end{array}\right)$, $\sigma_k$ acts as the transposition $(0,1)$ and $\sigma_{k+1}$ acts as the cycle $(0,\ldots, n_k-1)$ (see figure \ref{Ak}). Notice that there is no contraction for the automaton $A_1$.

%Les automates : $A_0,A_1, A_2,...,A_{\alpha}$ avec pour tout $A_k$, $\sigma_{k-1}:(0\leftarrow1)$, $\sigma_k:(0\leftrightarrow1)$, $\sigma_{k+1}:(0..n-1)$. (Remarque : il n'y a pas de contraction pour $A_0$.)

\begin{figure}[htb]
	\centerline{
		\begin{tikzpicture}[node distance=2cm, bend angle=25]
			\node[state,initial] (0) {$0$};
			\node[state] (1) [right of=0] {$1$};
			\node[state] (2) [right of=1] {$2$};
			\node (etc) [right of=2] {$\ldots$};
			\node[state, rounded rectangle] (s-2) [right of=etc] {$n_k-2$};
			\node[state, rounded rectangle, accepting] (s-1) [right of=s-2] {$n_k-1$};
			\path[->]
        (0) edge[bend left] node {$\sigma_{k+1},\sigma_k$} (1)
        (1) edge[bend left] node {$\sigma_{k+1}$} (2)
        (2) edge[bend left] node {$\sigma_{k+1}$} (etc)
        (etc) edge[bend left] node {$\sigma_{k+1}$} (s-2)
        (s-2) edge[bend left] node {$\sigma_{k+1}$} (s-1)
        (s-1) edge[out=-115, in=-95, looseness=.4] node[above] {$\sigma_{k+1}$} (0)
        (1) edge[bend left] node {$\sigma_k,\sigma_{k-1}$} (0)
			;
    \end{tikzpicture}}
  \caption{The DFA $A_k$ without the identity transitions}
	\label{Ak}
%  \label{DFA A3, A2, A1 for catenation}
\end{figure}

Let us recall that each state is associated to a sequence. We want to prove that, for our family of automata, the size of the minimal DFA for the multiple catenation is the number of valid sequences.
\begin{proposition}\label{séparation}
Two distinct states  $s=(S_1,\ldots ,S_{\alpha})$ and $s'=(S_1',\ldots,S_{\alpha}')$ are not equivalent.
\end{proposition}

\begin{proof} By induction on  $\mu=\mathrm{max}(\{k\in[1,\alpha]|S_k\neq S_k'\})$. If $k=\alpha$ then  there exists $j$ such that $q_j\in S_{\alpha}\oplus S_{\alpha}'$ and the word $(\sigma_{\alpha+1})^{n_{\alpha}-j}$ separates $s$ and $s'$. If $k<\alpha$ then  there exists $j$ such that $q_j\in S_k\oplus S_k'$. The word $(\sigma_{k+1})^{n_k-j}\sigma_k(\sigma_{k+2})^{n_{k+1}-1}$ sends $s$ and $s'$, respectively, to two states $t=(T_1,...,T_{\alpha})$ and $t'=(T_1',...,T_{\alpha}')$ such that $0\in T_{k+1}\oplus T_{k+1}'$. The states $t$ and $t'$ being non equivalent by the induction hypothesis, $s$ and $s'$ also are.
\end{proof}

We now investigate the accessibility problem. The main difficulty appears when we have to access some valid state $s$ by using a contraction. In such a situation, it is a bit technical to find the predecessor of $s$. Indeed, a contraction on a DFA implies a transposition and a permutation on the two previous DFAs with some possible disturbances due to the property \textit{P3} (when a final state is reached in some DFA, the initial state of the next one is also reached). To solve this difficulty, we need some technical lemmas.
Let $A_k$ be an automaton defined in table \ref{tableau}. For any state $q$ of $A_k$, 
$p\cdot q$ stands for $\sigma_{k+1}\cdot q$  and 
$t\cdot q$ stands for $\sigma_k\cdot q$. As usual, for any set of states  $S$, we denote  $p\cdot S=\bigcup_{q\in S}p\cdot q$ and $t\cdot S=\bigcup_{q\in S}t\cdot q$. These notations allows to shorten some expressions by omitting unambiguous indexes.
%\end{array}$
%\end{definition}
%
\begin{definition}
For any state $s=(i,S_2,\ldots ,S_{k-1},S_k,S_{k+1},\ldots ,S_{\alpha})$, we define the two following transformations:
\begin{itemize}
\item
		$\tau_k\cdot s=(i,S_2,\ldots ,S_{k-3},p\cdot S_{k-2},t\cdot S_{k-1},S_k\cup\{1\},S_{k+1},\ldots ,S_{\alpha})$, 
		$\forall k\in [2,\alpha]$\\
	%	$\tau_3\cdot s=(p\cdot i,e\cdot S_2,S_3\cup\{1\},S_4,\ldots ,S_{\alpha})$,\\
	%	$\tau_2\cdot s=(e\cdot i,S_2\cup \{1\},S_3,\ldots ,S_{\alpha})$.
\item
  	$\nu_k\cdot s=
  	(i,S_2,\ldots ,p^{\mathrm{min}(S_{k-3})+1}\cdot S_{k-3},p\cdot (S_{k-2}\setminus\{0\}),t\cdot S_{k-1},S_k\cup\{1\},S_{k+1},\ldots ,S_{\alpha})$, 
  	$\forall k\in [4,\alpha]$\\
  %	$\nu_4\cdot s=(q_f,p\cdot (S_2\setminus\{0\}),e\cdot S_3,S_4\cup\{1\},S_5,\ldots ,S_{\alpha})$,\\
  	$\nu_3\cdot s=\left\{
  	\begin{array}{ll}
  	  (p.i,t\cdot S_2,S_3\cup\{1\},S_4,\ldots ,S_{\alpha}) & \mathrm{if}\ i>0\\
  	  (0,t\cdot S_2,S_3\cup\{1\},S_4,\ldots ,S_{\alpha}) & \mathrm{if}\ i=0
  	\end{array}
  	\right.$
\end{itemize}
%where $j_0=\mathrm{min}(S_{k-3})$.
\end{definition}

\begin{lemma}\label{valid states}
Any valid state $s$ satisfies the  following properties:
\begin{enumerate}
\item $\forall k\in [2,\alpha]$, $\tau_k\cdot s$ is a valid state if and only if ($0\not\in S_{k-2}$ or $1\in S_{k-1}$) and
	($q_f\not\in S_{k-3}$ or $1\in S_{k-2}$).\hfill ($P_{\tau}$)
\item $\forall k\in [4,\alpha]$, $\nu_k\cdot s$ is a valid state if and only if  $1\in S_{k-2}$ and ($q_f\not\in S_{k-4}$ or
  $1\in S_{k-3}$).\hfill ($P_{\nu}$)
\item $\nu_3\cdot s$ is always a valid state.
\end{enumerate}
\end{lemma}
\begin{proof}
\begin{enumerate}
\item Property  \textit{P1} is clear from the definition of $\tau$. Property \textit{P2} is deduced  from the fact that  $s$ is a valid state and that the size of any set of $s$ can not be decreased by  $\tau_k$. Property \textit{P3} comes, for $k-1$, from the first  parenthesis which asserts that it is not possible to have simultaneously $q_f \in p\cdot S_{k-2}\ (0\in S_{k-2})$ and $0\not\in e\cdot S_{k-1}\ (1\not\in S_{k-1})$, for $k-2$, from the second parenthesis which asserts that it is not possible to have simultaneously $q_f\in S_{k-3}$ and $0\not\in p\cdot S_{k-2}\ (1\not\in S_{k-2})$, and for all the other sets, from the fact that $s$ is a valid state. 
\item Property  \textit{P1} is clear from the definition of $\nu$.  Property \textit{P2} is deduced  from the fact that  $s$ is a valid state and because $1\in S_{k-2}$ involves $S_{k-2}\neq\emptyset$. Property \textit{P3} comes, for  $k-3$, from the last parenthesis which asserts that it is not possible to have simultaneously  $q_f \in S_{k-4}$ (which implies $j_0=0$) and $0\not\in p^{j_0+1}\cdot S_{k-3}\ (1\not\in S_{k-3})$, for $k-2$, from the fact that $0\in p.(S_{k-2}\setminus\{0\})$ ($1\in S_{k-2}$), for $k-1$, from the fact that $q_f\not\in p\cdot (S_{k-2}\setminus\{0\})$ and, for all the other sets, from the fact that $s$ is a valid state. 
\item clear from the definition.
\end{enumerate}
\end{proof}
\begin{remark}\label{states}
The previous proof of properties $P_{\tau}$ and $P_{\nu}$ can be admitted for the first values of  $k$, if we accept the convention that each time a set $S_j$ is considered with $j<1$, this set is assimilated to $\emptyset$. With this convention, the properties can be simplified as:
%Properties $P_{\tau}$ and $P_{\nu}$ are checked  with similar arguments for the first values of  $k$, but their statement is simplified as follows:
\begin{itemize}
\item $\tau_3\cdot s$ is a valid state if and only if ($i\neq 0$ or $1\in S_2$).
\item $\tau_2\cdot s$ is always a valid state.
\item $\nu_4\cdot s$ is a valid state if and only if $1\in S_2$.
\end{itemize}
This convention is often implicitly used in the following.
\end{remark}

A composition of such transformations is denoted by a word over the alphabet $\Pi=\{\tau_k\}_{k\in[2,\alpha]}\cup\{\nu_k\}_{k\in[3,\alpha]}$.

\begin{lemma}\label{mot de retour}
For any valid state  $s=(i,S_2,\ldots,S_{k-1},S_k,S_{k+1},\ldots,S_{\alpha})$ such that $0\in S_k$ et $1\not\in S_k$ we have:
\begin{enumerate}
\item if $\tau_k\cdot s$ is a valid state then  $(\tau_k\cdot s)\cdot\sigma_{k-1}=s$.
\item if $\nu_k\cdot s$ is a valid state and ($0\in S_{k-2}$ and $1\not\in S_{k-1}$) then $(\nu_k\cdot
  s)\cdot\sigma_{k-1}(\sigma_{k-2})^{\mathrm{min}(S_{k-3})+1}=s$ (for $k>3$).
\item if $i>0$ then $(\nu_3\cdot s)\cdot\sigma_2=s$.
\item if $i=0$ and $1\not\in S_2$ then $(\nu_3\cdot s)\cdot\sigma_2\sigma_1=s$.
\end{enumerate}
\end{lemma}

\begin{proof}
\begin{enumerate}
\item It is clear from the definition of  $\tau_k$ that its action is cancelled  by a transition labelled 
  $\sigma_{k-1}$.
\item Let us first notice, as for $\tau_k$, that $\sigma_{k-1}$ cancels the action of $\nu_k$ on the sets $S_{k-2}$,
  $S_{k-1}$ and $S_k$. Then, $(\sigma_{k-2})^{\mathrm{min}(S_{k-3})+1}$ cancels the action of $\nu_k$ on $S_{k-3}$ and does not modify
  $S_{k-2}$ and $S_{k-1}$ if $1\in S_{k-2}$ (which is true according to ($P_{\nu}$)) and $1\not\in S_{k-1}$ (which is true by hypothesis).
\item If $i>0$, $\sigma_2$ cancels the action of $\nu_3$.
\item If $i=0$, $\sigma_2$ cancels the action of  $\nu_3$ on $S_3$ and $S_2$ but sends $i$ in $1$. The action of $\sigma_1$ allows us to send $i$ in $0$ without modifying $S_3$ and $S_2$ (because $1\not\in S_2$).
\end{enumerate}
\end{proof}

In the following lemma, let us write $\Sigma_k=\{\sigma_1,\ldots ,\sigma_k\}$.

\begin{lemma}\label{theLemme}
When $k\geq2$, for any valid state $s=(i,S_2,...,S_{k-1},S_k,S_{k+1},...,S_{\alpha})$ such that $0\in S_k$ and $1\not\in S_k$ there exists a couple $(u,v)\in(\Pi^*,(\Sigma_{k-2})^*)$ such that $s'=u\cdot s$ is a valid state of the form $(i',S_2',...,S_{k-1}',S_k\cup\{1\},S_{k+1},...,S_{\alpha})$ and $s'\cdot\sigma_{k-1}v=s$.
\end{lemma}

\begin{proof}
By induction on $k$. Multiple cases can appear. Someones, labeled (B), are bases cases. Other ones, labeled (IH), use the induction hypothesis.
\begin{itemize}[align=left]
\item If $1\in S_{k-1}$ we distinguish two cases:
	\begin{itemize}[align=left]
	\item[$\bullet$ (B1)] If $q_f\not\in S_{k-3}$ or $1\in S_{k-2}$, the desired couple $(u,v)$ is $(\tau_k,\varepsilon)$. Indeed,
	  $\tau_k\cdot s$ is a valid state by $(P_{\tau})$, of the announced form by definition of $\tau_k$ and $(\tau_k\cdot
	  s)\cdot\sigma_{k-1}=s$ by lemma \ref{mot de retour}.
	\item[$\bullet$ (IH1)] If $q_f\in S_{k-3}$ (which implies $0\in S_{k-2}$) and $1\not\in S_{k-2}$ then, by induction
	  hypothesis, there exists a couple $(u',v')\in(\Pi^*,(\Sigma_{k-4})^*)$ such that $s''=u'\cdot s$ is a valid state of the form
	  $(i'',S_2'',...,S_{k-3}'',S_{k-2}\cup\{1\},S_{k-1},...,S_{\alpha})$ and $s''\cdot\sigma_{k-3}v'=s$. Since $S_{k-1}''=S_{k-1}$ and $1\in
	  S_{k-2}''$, we are taken back to point (B1) and the desired couple $(u,v)$ is obtained by composing $(\tau_k,\varepsilon)$
	  with $(u',v')$, which gives $(\tau_ku',\sigma_{k-3}v')$.
	\end{itemize}
\item[\textbf{--} (B2)] If $1\not\in S_{k-1}$ and $0\not\in S_{k-2}$ (and so $q_f\not\in S_{k-3}$) then we follow a similar reasoning
  to the one used for case (B1) to find $(u,v)=(\tau_k,\varepsilon)$.
\item If $1\not\in S_{k-1}$ and $0,1\in S_{k-2}$ we distinguish two cases:
	\begin{itemize}[align=left]
	\item[$\bullet$ (B3)] If $q_f\not\in S_{k-4}$ or $1\in S_{k-3}$ then the desired couple $(u,v)$ is
	  $(\nu_k,(\sigma_{k-2})^{j_0+1})$, where
		\[
		  j_0 = \left\{
	    \begin{array}{ll}
		    \mathrm{min}(S_{k-3}) & \ \mathrm{if}\ k>3\\
		    -1                    & \ \mathrm{if}\ k=3\ \mathrm{and}\ i>0\\
		    0                     & \ \mathrm{if}\ k=3\ \mathrm{and}\ i=0
	    \end{array}
	    \right.
	  \]
	  Indeed, $\nu_k\cdot s$ is a state by $(P_{\nu})$, of the announced form by the definition of $\nu_k$ and
	  $(\nu_k\cdot s)\cdot\sigma_{k-1}(\sigma_{k-2})^{j_0+1}=s$ by lemma \ref{mot de retour}.
	\item[$\bullet$ (IH2)] If $q_f\in S_{k-4}$ (which implies $0\in S_{k-3}$) and $1\not\in S_{k-3}$ then, by induction
	  hypothesis, there exists a couple $(u',v')\in(\Pi^*,(\Sigma_{k-5})^*)$ such that $s''=u'\cdot s$ is a valid state of the form
	  $(i'',S_2'',...,S_{k-4}'',S_{k-3}\cup\{1\},S_{k-2},...,S_{\alpha})$ and $s''\cdot\sigma_{k-4}v'=s$. Since
	  $S_{k-1}''=S_{k-1}$, $S_{k-2}''=S_{k-2}$ and $1\in S_{k-3}''$, we are taken back to point (B3) (with $j_0=0$) and the desired 
		couple $(u,v)$ is obtained by composing $(\nu_k,\sigma_{k-2})$ with $(u',v')$, which gives
	  $(\nu_ku',\sigma_{k-2}\sigma_{k-4}v')$.
	\end{itemize}
\item If $1\not\in S_{k-1}$, $0\in S_{k-2}$ and $1\not\in S_{k-2}$, we distinguish two cases:
	\begin{itemize}[align=left]
	\item[$\bullet$ (IH3)] If $k>3$ then by induction hypothesis, there exists a couple
	  $(u',v')\in(\Pi^*,(\Sigma_{k-4})^*)$ such that $s''=u'\cdot s$ is a valid state of the form
	  $(i'',S_2'',...,S_{k-3}'',S_{k-2}\cup\{1\},S_{k-1},...,S_{\alpha})$ and $s''\cdot\sigma_{k-3}v'=s$ (remark: if
	  $k=4$ then $v'=\varepsilon$). Since $S_{k-1}''=S_{k-1}$ and $0,1\in S_{k-2}''$, we are taken back to previous cases
	  ((B3) or (IH2) according to $S_{k-3}$ and $S_{k-4}$) which allow to find a couple $(u'',v'')$ such that $u''\cdot s''$
	  is a valid state and $(u''\cdot s'')\cdot\sigma_{k-1}v''=s''$. The desired couple $(u,v)$ is obtained by composing
	  $(u'',v'')\in(\Pi^*,(\Sigma_{k-2})^*)$ with $(u',v')$, which gives $(u''u',v''\sigma_{k-3}v')$.
	\item [$\bullet$ (B4)] If $k=3$ then the desired couple $(u,v)$ is $(\nu_3,\varepsilon)$ or $(\nu_3,\sigma_1)$ depending on whether
	  $i>0$ or $i=0$. Indeed, $\nu_3\cdot s$ is always a valid state (by lemma \ref{valid states}) of the announced form
	  by definition of $\nu_3$ and $(\nu_3\cdot s)\cdot\sigma_2v=s$ by lemma \ref{mot de retour} (applicable, because $1\not\in S_{k-1}$,
		i.e. $1\not\in S_2$).
	\end{itemize}
\end{itemize}
\end{proof}

\begin{corollary}\label{theCorollaire}
For any valid state $s=(i,S_2,...,S_{\alpha-1},\{0\})$ such that $q_f\not\in S_{\alpha-1}$ there exists a valid state $s'=(i',S_2',...,S_{\alpha-1}',\{1\})$ and a word $w$ such that $s'.\sigma_{\alpha-1}w=s$.
\end{corollary}

\begin{proof}
By previous lemma, there exists a state $s''=(i'',S_2'',...,S_{\alpha-1}'',\{0,1\})$ and a word $v\in(\Sigma_{\alpha-2})^*$ such that $s''\cdot\sigma_{\alpha-1}v=s$. It is easy to see, because of the alphabet of $v$, that $q_f\in S_{\alpha-1}''$ if and only if $q_f\in S_{\alpha-1}$. Let us denote $s'''$ the valid state such that $s''\cdot\sigma_{\alpha-1}=s'''$. Since $q_f\not\in S_{\alpha-1}''$, the sequence $(i'',S_2'',...,S_{\alpha-1}'',\{1\})$ is a valid state. Furthermore, we notice that this state verifies: $(i'',S_2'',...,S_{\alpha-1}'',\{1\}).\sigma_{\alpha-1}=s'''$. So, we deduce the desired state is $s'=(i'',S_2'',...,S_{\alpha-1}'',\{1\})$ and $w=v$.
\end{proof}

\begin{proposition}\label{accessibilité}
Any valid state $s=(i,S_2,...,S_{\alpha})$ is accessible.
\end{proposition}

\begin{proof}
By induction on $\alpha$, the base case being when $\alpha=1$. In this case, the proposition is trivially verified since $A_1$ is a minimal DFA. Now, by induction hypothesis, we know that each valid state of the form $t=(i,S_2,...,S_{\alpha-1})$ is accessible. As we will see, if $q_f\in S_{\alpha-1}$, one can suppose verified the property (${\cal P}$) stating that $t$ is reached in the following way: $(0,\emptyset,...,\emptyset)\xrightarrow{\lambda_1}t_1\xrightarrow{\lambda_2}t_2 ...\xrightarrow{\sigma_{\alpha}}t_k ...\xrightarrow{\lambda_{\ell}}t$ with $\forall {t_i}_{(i<k)}, q_f\not\in S_{\alpha}$ et $\forall j>k, \lambda_j\neq\sigma_{\alpha}$ (intuitively, when we reach the final state of $A_{\alpha-1}$ we no longer permute on this automaton). So, the induction hypothesis allows to suppose accessible, any valid state of the form $(i,S_2,...,S_{\alpha-1},\emptyset)$, as well as any valid state of the form $(i,S_2,...,S_{\alpha-1},\{0\})$ with $q_f\in S_{\alpha-1}$. To prove the accessibility of $s$, we follow a second induction based on the following partial order, defined on the possible sets $S_{\alpha}=\{k_0,k_1,...\}$:
\begin{tabbing}
	$S_{\alpha}<S_{\alpha}'$ if \=$|S_{\alpha}|<|S_{\alpha}'|$ or\\
	 	                          \>$|S_{\alpha}|=|S_{\alpha}'|$ and $k_0<k_0'$ or\\
	  	                        \>$|S_{\alpha}|=|S_{\alpha}'|$, $k_0=k_0'$ and $k_1<k_1'$.
\end{tabbing}
We first prove that each state is accessible when $S_{\alpha}$ is a singleton. We set $S_{\alpha-1}=\{j_0,j_1,...\}$ and $S_{\alpha}=\{k_0\}$.
\begin{itemize}
\item If $j_0>0$ then $q_f\not\in S_{\alpha-2}$. The state
  $s'=(i,S_2,\ldots,S_{\alpha-2},(\sigma_{\alpha})^{j_0+1}\cdot S_{\alpha-1},\{0\})$ is accessible by induction hypothesis (because $q_f\in
  (\sigma_{\alpha})^{j_0+1}.S_{\alpha-1}$). It is easy to verify that $s'$ is valid (mainly because $q_f\not\in S_{\alpha-2}$).
	If $j_0$ is odd then $s'\xrightarrow{(\sigma_{\alpha})^{j_0+1}(\sigma_{\alpha+1})^{k_0}}s$. If $j_0$ is even then
  $s'\xrightarrow{(\sigma_{\alpha})^2\sigma_{\alpha+1}(\sigma_{\alpha})^{j_0-1}(\sigma_{\alpha+1})^{k_0}}s$.
\item If $j_0=0$ and $k_0>0$, we distinguish two cases:
	\begin{itemize}
	\item If $q_f\not\in S_{\alpha-2}$ or $j_1=1$ then the state $s'=(i,S_2,\ldots,S_{\alpha-2},\sigma_{\alpha}\cdot S_{\alpha-1},\{0\})$ is
	  valid and accessible by induction hypothesis (because $q_f\in \sigma_{\alpha}\cdot S_{\alpha-1}$) and
	  $s'\xrightarrow{\sigma_{\alpha}(\sigma_{\alpha+1})^{k_0-1}}s$.
  \item If $q_f\in S_{\alpha-2}$ (and so $0\in S_{\alpha-1}$) and $1\not\in S_{\alpha-1}$ then by Lemma
    \ref{theLemme}, there exists a valid state $s'=(i',S_2',\ldots,S_{\alpha-2}',S_{\alpha-1}\cup\{1\},\{k_0\})$ and a word $v$
		such that $s'\xrightarrow{v}s$. And $s'$ is accessible following the previous point.
	\end{itemize}
\item If $j_0=0$ and $k_0=0$ then, either $q_f\in S_{\alpha-1}$ and $s$ is accessible by induction hypothesis, or
  $q_f\not\in S_{\alpha-1}$ and, by Corollary \ref{theCorollaire}, there exists a valid state
  $s'=(i',S_2',\ldots,S_{\alpha-1}',\{1\})$ and a word $w$ such that $s'\xrightarrow{w}s$. And $s'$ is
  accessible following one of the previous points.
\end{itemize}
Now, we look at the case where $S_{\alpha}$ contains at least two states.
\begin{itemize}
\item If $k_0>0$ then the state $s'=(i,S_2,\ldots,S_{\alpha-1},\sigma_{\alpha+1}\cdot S_{\alpha})$ is valid and accessible by induction hypothesis
  (we have decreased by $1$ the first index of $S_{\alpha}$) and $s'\cdot\sigma_{\alpha+1}=s$.
\item If $k_0=0$ and $k_1=1$, we distinguish two cases:
	\begin{itemize}
	\item If $q_f\not\in S_{\alpha-2}$ or $1\in S_{\alpha-1}$ then the state
	  $s'=(i,S_2,\ldots,S_{\alpha-2},(\sigma_{\alpha})^{j_0+1}\cdot
		S_{\alpha-1},\sigma_{\alpha+1}\cdot(S_{\alpha}\setminus\{0\}))$ is valid and accessible by induction hypothesis (the last
		set contains one less state than $S_{\alpha}$) and $s'\xrightarrow{\sigma_{\alpha+1}(\sigma_{\alpha})^{j_0+1}}s$.
	\item If $q_f\in S_{\alpha-2}$ (which implies $0\in S_{\alpha-1}$) and $1\not\in S_{\alpha-1}$ then, by Lemma \ref{theLemme},
	  there exists a valid state $s'=(i',S_2',\ldots,S_{\alpha-2}',S_{\alpha-1}\cup\{1\},S_{\alpha})$ and a word $w$ such that
    $s'\xrightarrow{w}s$. And $s'$ is accessible following the previous point. 
	\end{itemize}
\item If $k_0=0$ and $k_1>1$ then, by Lemma \ref{theLemme}, there exists a valid state
    $s'=(i',S_2',\ldots,S_{\alpha-1}',S_{\alpha}\cup\{1\})$ and a word $w$ such that
    $s'\xrightarrow{w}s$. And $s'$ is accessible following the previous case.
\end{itemize}
One can verify that, in each of the considered cases, we never act in a final valid state  (\textit{i.e.} ($S_1,\ldots, S_\alpha)$ with $q_f\in S_\alpha$)  with a  $\sigma_{\alpha+1}$ letter. This ensures the property ${\cal P}$ announced at the beginning of the proof.
\end{proof}

From Propositions \ref{accessibilité} and \ref{séparation}, we deduce:

\begin{theorem}
The family of sequences of minimal DFAs $(A_1,\ldots,A_\alpha)_{\alpha>0}$, described in table \ref{tableau}, is a family of witnesses over an ($\alpha+1$)-letters alphabet for the catenation of $\alpha$ languages.
\end{theorem}
% !TEX root = main.tex
%\section{The case of $ABC$}

\section{A $\alpha$-letters witness for the catenation of $\alpha$ automata: a conjecture}

In this section we propose to decrease by one the size of the alphabet used to define witnesses for multiple catenation. A $\alpha$-letters alphabet should be optimal. In any case, it is optimal when $\alpha=2$: indeed it is proven in \cite{YZS94} that state complexity for catenation of two minimal DFAs with size $m$ and $n$, and using only one letter is $mn$, which is strictly lower to the general state complexity for catenation.

Our statement is a conjecture, since we only prove it when $\alpha=2$ and $\alpha=3$. Some tests computed with the software Sage for $\alpha\in[2,7]$ and DFAs with size in $[3,6]$ also argue in this sense. Our witnesses can be obtained by slightly modifying the table of the previous section (see Table \ref{tableauBis}).

\begin{table}[h]
\centerline{  
\begin{tikzpicture}[scale=0.5,node distance=1cm,shorten >=0pt]   
  \draw (1,1)rectangle (9,9);
%	      \foreach \x in {1,...,5} {
%	        \foreach \y in {1,...,6} {
%	          \pgfmathparse{\x+1} \let\z\pgfmathresult
%	          \pgfmathparse{\y+1} \let\t\pgfmathresult
%	          \draw[fill=gray!40] (\x,\y) rectangle (\x+1,\y+1);
%	        }
%	      }  
%	      \foreach \x/\y in {2/3, 3/1 ,4/1 ,4/4} {
%	        \pgfmathparse{\x+1} \let\z\pgfmathresult
%	        \pgfmathparse{\y+1} \let\t\pgfmathresult
%	        \draw[fill=white] (\x,\y) rectangle (\x+1,\y+1);
%	        \draw (\x,\y) -- (\z,\t);
%	        \draw (\z,\y) -- (\x,\t);
%	      }  	 
	\draw (1.5,9.5) node (q0) {$A_1$};
	\draw (2.5,9.5) node (q1) {$A_2$};
	\draw (3.5,9.5) node (q2) {$A_3$};
	\draw (6,9.5) node (q3) {$\ldots$};
%	\draw[dotted] (4.5,10.5) -- (7.5,10.5);
	\draw (8.5,9.5) node (q4) {$A_\alpha$};
	\draw (0.2,8.5) node (r0) {$\sigma_1$};
	\draw (0.2,7.5) node (r1) {$\sigma_2$};
	\draw (0.2,6.5) node (r2) {$\sigma_3$};
	\draw (0.2,4) node (r3) {$\vdots$};
	\draw (0.2,1.5) node (r4) {$\sigma_\alpha$};
%	\draw (0.2,1.5) node (r5) {$\sigma_{\alpha+1}$};
%
	\draw (1.5,8.5) node (q0) {$t$};
	\draw (2.5,8.5) node (q1) {$c$};
	\draw (3.5,8.5) node (q2) {$\mathds{1}$};
%	\draw (4.5,9.5) node (q3) {$i$};
%	\draw (6.5,9.5) node (q4) {$\ldots$};
	\draw (3.8,8.5) -- (8.2,8.5);
	\draw (8.5,8.5) node (q4) {$\mathds{1}$};
	\draw (1.5,7.5) node (q0) {$p$};
	\draw (2.5,7.5) node (q1) {$t$};
	\draw (3.5,7.5) node (q2) {$c$};
%	\draw (4.5,8.5) node (q3) {$i$};
%	\draw (6.5,8.5) node (q4) {$\ldots$};
%	\draw (8.5,8.5) node (q4) {$i$};
%
	\draw (1.5,6.5) node (q0) {$\mathds{1}$};
	\draw (2.5,6.5) node (q1) {$p$};
	\draw (3.5,6.5) node (q2) {$t$};
	\draw (4.5,6.5) node (q3) {$c$};
%	\draw (6.5,7.5) node (q4) {$\ldots$};
%	\draw (8.5,7.5) node (q4) {$i$};
%
  \draw [dashed] (4,6) -- (7,3);
  \draw (8.5,1.5) node (q0) {$p$};
  \draw (7.5,2.5) node (q0) {$t$};
  \draw [dashed] (3,6) -- (7,2);
%\draw (8.5,1.5) node (q0) {$p$};
  \draw (7.5,1.5) node (q0) {$p$};
  \draw [dashed] (5,6) -- (7,4);
  \draw (8.5,2.5) node (q0) {$c$};
  \draw (7.5,3.5) node (q0) {$c$};
  \draw  (1.5,6.2) -- (1.5,1.8);
  \draw  (1.8,1.5) -- (6.2,1.5);
  \draw (1.8,6.2) -- (6.2,1.8);
  \draw (1.5,1.5) node (q0) {$\mathds{1}$};
  \draw (6.5,1.5) node (q0) {$\mathds{1}$};
  \draw (8.5,3.5) node (q0) {$\mathds{1}$};
  \draw  (3.8,8.2) -- (8.2,3.8);
  \draw  (8.5,8.2) -- (8.5,3.8);
  \draw (3,3) node (q0) {$\mathds{1}$};
  \draw (7,7) node (q0) {$\mathds{1}$};
\end{tikzpicture}}
\caption{$\alpha$-letters witnesses for the multiple catenation: a conjecture}
\label{tableauBis}
\end{table}  

\subsection{The two automata case}
In its paper \cite{Brz13}, where J. Brzozowski proposes four atomic constructions to build universal witnesses, he observes a defect concerning the operation of catenation. He only suggests a 3-letters alphabet witness, whereas, in \cite{Jir05}, G. Jiraskova produces a 2-letters one.
%In \cite{Brz13} J. Brzozowski proposes four atomic constructions to build universal witnesses. It seems its conjecture works well in many situations (he gives 19 conjectures supported by a software tool in this sense). However, he observes a defect concerning the operation of catenation. He only suggests a 3-letters alphabet witness, whereas, in \cite{Jir05}, G. Jiraskova produces a 2-letters one.

%It is significant to observe that the two automata given by Jiraskova are strongly dependent each other. Especially, the sink of a transition in the first automaton depends on the states number of the second one. 

%Conversely the approach of Brzozowski allows to have more independent DFA (in particular, they can be build in any order).
We give here a 2-letters witness for catenation, based on the atomic constructions of J. Brzozowski, which corresponds to the previous table when $\alpha=2$ (see Table \ref{two-letters witness} and Figure  \ref{2-letters}). 

\begin{table}
%\begin{minipage}{0.25\linewidth}
\centerline{\begin{tikzpicture}[scale=0.7,node distance=1cm,shorten >=0pt]   
  \draw (3,1)rectangle (5,3);
	\draw (3.5,3.5) node (q0) {$A_1$};
	\draw (4.5,3.5) node (q1) {$A_2$};
	\draw (2.5,2.5) node (q0) {$b$};
	\draw (2.5,1.5) node (q1) {$a$};
	\draw (3.5,2.5) node (q0) {$t$};
	\draw (4.5,2.5) node (q1) {$c$};
	\draw (3.5,1.5) node (q0) {$p$};
	\draw (4.5,1.5) node (q1) {$p$};
\end{tikzpicture}}
\caption{{$2$-letters witness for catenation of two languages}}
\label{two-letters witness}
%\end{minipage}
\end{table}

%The two DFAs corresponding to the witness of Table \ref{two-letters witness}  are described in Figure \ref{2-letters}.

\begin{figure}
%\begin{minipage}{0.74\linewidth}
%\begin{figure}[htb]
	\centerline{
		\begin{tikzpicture}[node distance=1.2cm, bend angle=25]
			\node[state,initial] (p0) {$p_0$};
			\node[state] (p1) [right of=p0] {$p_1$};
			\node[state] (p2) [right of=p1] {$p_2$};
			\node (etc2) [right of=p2] {$\ldots$};
		%	\node[state, rounded rectangle] (m-3) [right of=etc2] {$p_{m-3}$};
			\node[state, rounded rectangle] (m-2) [right of=etc2] {$p_{m-2}$};
			\node[state, rounded rectangle, accepting] (m-1) [right of=m-2] {$p_{m-1}$};
			\path[->]
        (p0) edge[bend left] node {a,b} (p1)
        (p1) edge[bend left] node {a} (p2)
        (p2) edge[bend left] node {a} (etc2)
        (etc2) edge[bend left] node {a} (m-2)
     %   (m-3) edge[bend left] node {a} (m-2)
        (m-2) edge[bend left] node {a} (m-1)
        (m-1) edge[out=-115, in=-65, looseness=.2] node[above] {a} (p0)
		    (p2) edge[out=115,in=65,loop] node {b} (p2)
	%	    (m-3) edge[out=115,in=65,loop] node {b} (m-3)
		    (m-2) edge[out=115,in=65,loop] node {b} (m-2)
		    (m-1) edge[out=115,in=65,loop] node {b} (m-1)
        (p1) edge[bend left] node[above] {b} (p0)
			;
			\node (concon) [right of=m-1]{};
			\node[state,initial] (q0)  [right of=concon]{$q_0$};
			\node[state] (q1) [right of=q0] {$q_1$};
			\node[state] (q2) [right of=q1] {$q_2$};
			\node (etc2) [right of=q2] {$\ldots$};
			%\node[state, rounded rectangle] (n-3) [right of=etc2] {$q_{n-3}$};
			\node[state, rounded rectangle] (n-2) [right of=etc2] {$q_{n-2}$};
			\node[state, rounded rectangle, accepting] (n-1) [right of=n-2] {$q_{n-1}$};
			\path[->]
        (q0) edge[bend left] node {a} (q1)
        (q1) edge[bend left] node {a} (q2)
        (q2) edge[bend left] node {a} (etc2)
        (etc2) edge[bend left] node {a} (n-2)
%        (n-3) edge[bend left] node {a} (n-2)
        (n-2) edge[bend left] node {a} (n-1)
        (n-1) edge[out=-115, in=-65, looseness=.2] node[above] {a} (q0)
		    (q0) edge[out=115,in=65,loop] node {b} (q0)
		    (q2) edge[out=115,in=65,loop] node {b} (q2)
		  %  (n-3) edge[out=115,in=65,loop] node {b} (n-3)
		    (n-2) edge[out=115,in=65,loop] node {b} (n-2)
		    (n-1) edge[out=115,in=65,loop] node {b} (n-1)
        (q1) edge[bend left] node[above] {b} (q0)
			;
						%\draw (5.5,-1.8) node (q) {DFA $B$};
    \end{tikzpicture}
  }
%  \end{minipage}
  \caption{The witness described by Table \ref{two-letters witness}.}
	\label{2-letters}
\end{figure}
%\begin{figure}[htb]
	%\centerline{
		%\begin{tikzpicture}[node distance=1.5cm, bend angle=25]
			%\node[state,initial] (q0) {$q_0$};
			%\node[state] (q1) [right of=q0] {$q_1$};
			%\node[state] (q2) [right of=q1] {$q_2$};
			%\node (etc2) [right of=q2] {$\ldots$};
			%\node[state, rounded rectangle] (n-3) [right of=etc2] {$q_{n-3}$};
			%\node[state, rounded rectangle] (n-2) [right of=n-3] {$q_{n-2}$};
			%\node[state, rounded rectangle, accepting] (n-1) [right of=n-2] {$q_{n-1}$};
			%\path[->]
        %(q0) edge[bend left] node {a} (q1)
        %(q1) edge[bend left] node {a} (q2)
        %(q2) edge[bend left] node {a} (etc2)
        %(etc2) edge[bend left] node {a} (n-3)
        %(n-3) edge[bend left] node {a} (n-2)
        %(n-2) edge[bend left] node {a} (n-1)
        %(n-1) edge[out=-115, in=-65, looseness=.2] node[above] {a} (q0)
		    %(q0) edge[out=115,in=65,loop] node {b} (q0)
		    %(q2) edge[out=115,in=65,loop] node {b} (q2)
		    %(n-3) edge[out=115,in=65,loop] node {b} (n-3)
		    %(n-2) edge[out=115,in=65,loop] node {b} (n-2)
		    %(n-1) edge[out=115,in=65,loop] node {b} (n-1)
        %(q1) edge[bend left] node[above] {b} (q0)
			%;
    %\end{tikzpicture}
  %}
  %\caption{AFD B}
  %\label{AFD B for catenation}
%\end{figure}

Following Definition \ref{autoConcat}, we add transitions from the predecessor of the final state in the first DFA to the initial state of the second DFA and apply the subset construction to the resulting NFA. % as described in Definition \ref{}.
%If we add an $\varepsilon$ transition from $p_{m-1}$ to $q_0$ we obtain a NFA recognizing $L(A)L(B)$. Applying the subset construction,
The valid states of this automaton, named $A$ in the following, are of the form $(p_i,S)$, where $S$ denotes any subset of $\{q_0,...,q_{n-1}\}$ (containing $q_0$ if $i=m-1$). The number of valid states  is equal to the state complexity of catenation, that is $m2^n-2^{n-1}$. We prove that all these states are both accessible and pairwise non-equivalent.

\begin{proposition}\label{2-accessibility}
Each valid state $(p_i,S)$ of $A$  is accessible.
\end{proposition}

\begin{proof}
By induction on the size of $S$. First, any state of the form $(p_i,\emptyset)$  with $i\neq m-1$ is accessible from the initial state $(p_0,\emptyset)$ by the word $a^i$.

Now, consider any state $s=(p_i,S)$ with $|S|=k$ for some integer $k>0$. We proceed by cases :
\begin{enumerate}
	\item\label{point1} If $i=m-1$ then $q_0\in S$ and $s$ is reached by $a$ from $(p_{m-2},a.(S\backslash\{q_0\}))$, which is accessible by induction hypothesis.
	\item\label{point2} If $i=0$ and $q_1\in S$,  $s$ is reached by $a$ from $(m-1,a.S)$ which is accessible by point \ref{point1}.
	%either $S=\{q_{1},q_{j_2}...,q_{j_\alpha}\}$ or $S=\{q_{0},q_{1},q_{j_2}...,q_{j_\alpha}\}$ then $s$ is reached by $a$ from $(m-1,a.S)$ which is accessible by point \ref{point1}.
	\item\label{point3} If $i=0$ and $q_1\not\in S$. Let us set $q_1<\ldots<q_{n-1}<q_n=q_0$ and $S$ be an ordered set with $S=\{q_{j_1},...,q_{j_\alpha}\}$. Then $s$ is reached by $(ab)^{j_1-1}$ from $(0,\{q_1,...,q_{j_\alpha-j_1+1}\})$ which is accessible by point \ref{point2}.
	%S=\{q_{j_1},...,q_{j_\alpha}\}$ with $j_1>1$ then $s$ is reached by $(ab)^{j_1-1}$ from $(0,\{q_1,...,q_{j_\alpha-j_1+1}\})$ which is accessible by point \ref{point2}.
%	\item\label{point4} If $i=0$ and $S=\{q_0,q_{j_2},...,q_{j_\alpha}\}$ then $s$ is reached by $(ab)^{j_2-1}$ from $(0,\{q_1,...,q_{j_\alpha-j_2+1},q_{n-j_2+1}\})$ which is accessible by point \ref{point2}. (Remark: the indexes are modulo $n$. So, if $|S|=1$ then $j_2=0$ and $j_2-1=n-1$.)
	\item\label{point5} If $i\in]0,m-1[$ then $s$ is reached by $a^i$ from $(0,a^i.S)$ which is accessible by one of the  previous two points.
\end{enumerate}
\end{proof}

\begin{proposition}\label{2-separation}
Any two distinct  valid states $s=(p_i,S)$ and $s'=(p_{i'},S')$ of $A$ are non-equivalent.
\end{proposition}

\begin{proof}
There are two cases to consider:
\begin{itemize}
	\item If $S\neq S'$, without loss of generality, let $q_j\in S\setminus S'$. Then $a^{n-1-j}$ sends $s$ to a final state and $s'$ to a non-final one.
	\item Now, suppose $S=S'$. So $i\neq i'$. By reading the word $a^{(m-n-i)\bmod m}$, one sends $s$ to $(p_{i_1},S_1)$ with $i_1=(m-n)\bmod m$, and $s'$ to $(p_{i_1'},S_1')$ with $i_1'\neq i_1$. Then, by the word $bb$, we send $(p_{i_1},S_1)$ to $(p_{i_1},S_1\setminus\{q_1\})$ and $(p_{i_1'},S_1')$ to $(p_{i_1'},S_1'\setminus\{q_1\})$. Last, by reading the word $a^{n-1}$, we send $(p_{i_1},S_1\setminus\{q_1\})$ to $(p_{m-1},S_2)$ and $(p_{i_1'},S_1'\setminus\{q_1\})$ to $(p_{i_2'},S_2')$ with $q_0\in S_2\setminus S_2'$. So we have reduced this case to the previous one.
\end{itemize}
\end{proof}

It follows from Propositions \ref{2-accessibility} and \ref{2-separation} that:

\begin{theorem}
The couple of Brzozowski automata defined in Table \ref{two-letters witness} is a $2$-letters  witness for the catenation of two languages.
\end{theorem}

\subsection{The three automata case}

The triple of Brzozowski automata $A_1, A_2, A_3$ with respective size $m, n, p$ described in Figure \ref{3-letters} is the  proposed $3$-letters witness for the double catenation.
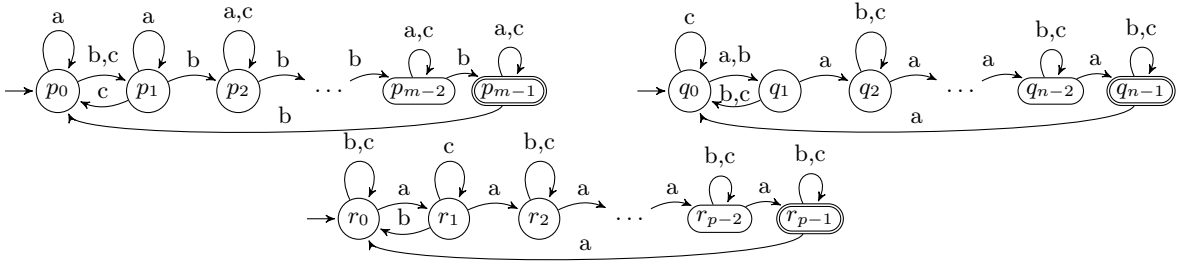
\begin{figure}[htb]
	\centerline{
		\begin{tikzpicture}[node distance=1.2cm, bend angle=25]
			\node[state,initial] (p0) {$p_0$};
			\node[state] (p1) [right of=p0] {$p_1$};
			\node[state] (p2) [right of=p1] {$p_2$};
			\node (etc1) [right of=p2] {$\ldots$};
			\node[state, rounded rectangle] (m-2) [right of=etc1] {$p_{m-2}$};
			\node[state, rounded rectangle, accepting] (m-1) [right of=m-2] {$p_{m-1}$};
			\path[->]
        (p0) edge[bend left] node {b,c} (p1)
        (p1) edge[bend left] node {b} (p2)
        (p2) edge[bend left] node {b} (etc1)
        (etc1) edge[bend left] node {b} (m-2)
        (m-2) edge[bend left] node {b} (m-1)
        (m-1) edge[out=-115, in=-65, looseness=.2] node[above] {b} (p0)
		    (p0) edge[out=115,in=65,loop] node {a} (p0)
		    (p1) edge[out=115,in=65,loop] node {a} (p1)
		    (p2) edge[out=115,in=65,loop] node {a,c} (p2)
		    (m-2) edge[out=115,in=65,loop] node {a,c} (m-2)
		    (m-1) edge[out=115,in=65,loop] node {a,c} (m-1)
        (p1) edge[bend left] node[above] {c} (p0)
			;
			\node (concon) [right of=m-1]{};
			\node[state,initial] (q0) [right of=concon]{$q_0$};
			\node[state] (q1) [right of=q0] {$q_1$};
			\node[state] (q2) [right of=q1] {$q_2$};
			\node (etc2) [right of=q2] {$\ldots$};
			\node[state, rounded rectangle] (n-2) [right of=etc2] {$q_{n-2}$};
			\node[state, rounded rectangle, accepting] (n-1) [right of=n-2] {$q_{n-1}$};
			\path[->]
        (q0) edge[bend left] node {a,b} (q1)
        (q1) edge[bend left] node {a} (q2)
        (q2) edge[bend left] node {a} (etc2)
        (etc2) edge[bend left] node {a} (n-2)
        (n-2) edge[bend left] node {a} (n-1)
        (n-1) edge[out=-115, in=-65, looseness=.2] node[above] {a} (q0)
		    (q0) edge[out=115,in=65,loop] node {c} (q0)
		    (q2) edge[out=115,in=65,loop] node {b,c} (q2)
		    (n-2) edge[out=115,in=65,loop] node {b,c} (n-2)
		    (n-1) edge[out=115,in=65,loop] node {b,c} (n-1)
        (q1) edge[bend left] node[above=-.1cm] {b,c} (q0)
			;
			\node (concon2) [below of=etc1, node distance=1.7cm]{};
			\node[state,initial] (r0) [right of=concon2, node distance=.4cm]{$r_0$};
			\node[state] (r1) [right of=r0] {$r_1$};
			\node[state] (r2) [right of=r1] {$r_2$};
			\node (etc3) [right of=r2] {$\ldots$};
			\node[state, rounded rectangle] (p-2) [right of=etc3] {$r_{p-2}$};
			\node[state, rounded rectangle, accepting] (p-1) [right of=p-2] {$r_{p-1}$};
			\path[->]
        (r0) edge[bend left] node {a} (r1)
        (r1) edge[bend left] node {a} (r2)
        (r2) edge[bend left] node {a} (etc3)
        (etc3) edge[bend left] node {a} (p-2)
        (p-2) edge[bend left] node {a} (p-1)
        (p-1) edge[out=-115, in=-65, looseness=.2] node[above] {a} (r0)
		    (r0) edge[out=115,in=65,loop] node {b,c} (r0)
		    (r1) edge[out=115,in=65,loop] node {c} (r1)
		    (r2) edge[out=115,in=65,loop] node {b,c} (r2)
		    (p-2) edge[out=115,in=65,loop] node {b,c} (p-2)
		    (p-1) edge[out=115,in=65,loop] node {b,c} (p-1)
        (r1) edge[bend left] node[above] {b} (r0)
			;
    \end{tikzpicture}
  }
  \caption{$3$-letters witness for double catenation}
  \label{3-letters}
\end{figure}

%\begin{itemize}
	%\item For $A_1$, $a:\mathds{1}$, $b:(0,\ldots, m-1)$, $c:(0,1)$.
	%\item For $A_2$, $a:(0,\ldots, n-1)$, $b:(0,1)$, $c:\left(\begin{array}{l}1\\0\end{array}\right)$.
	%\item For $A_3$, $a:(0,\ldots, p-1)$, $b:\left(\begin{array}{l}1\\0\end{array}\right)$, $c:\mathds{1}$.
%\end{itemize}
It corresponds to  Table  \ref{tableauBis} when $\alpha=3$ (see Table \ref{three-letters witness}).

\begin{table}
%\begin{minipage}{0.25\linewidth}
\centerline{\begin{tikzpicture}[scale=0.7,node distance=1cm,shorten >=0pt]   
  \draw (3,1) rectangle (6,4);
	\draw (3.5,4.5) node {$A_1$};
	\draw (4.5,4.5) node {$A_2$};
	\draw (5.5,4.5) node {$A_3$};
	\draw (2.5,3.5) node {$c$};
	\draw (2.5,2.5) node {$b$};
	\draw (2.5,1.5) node {$a$};
	\draw (3.5,3.5) node {$t$};
	\draw (4.5,3.5) node {$c$};
	\draw (5.5,3.5) node {$\mathds{1}$};
	\draw (3.5,2.5) node {$p$};
	\draw (4.5,2.5) node {$t$};
	\draw (5.5,2.5) node {$c$};
	\draw (3.5,1.5) node {$\mathds{1}$};
	\draw (4.5,1.5) node {$p$};
	\draw (5.5,1.5) node {$p$};
\end{tikzpicture}}
\caption{{$3$-letters witness for catenation of three languages}}
\label{three-letters witness}
%\end{minipage}
\end{table}

The accessible states of $A$ (the DFA obtained by the subset algorithm from $A_1$, $A_2$ and $A_3$ connected as described in Definition \ref{autoConcat}) are identified to 3-tuples of the form $(p_i,S=\{q_{j_0},...,q_{j_{\beta}}\},T=\{r_{k_0},...,r_{k_{\gamma}}\})$, with $j_0<...<j_{\beta}$ and $k_0<...<k_{\gamma}$, and must satisfy the three following constraints:
\begin{itemize}
	\item $i=m-1\Rightarrow q_0\in S$.
	\item $q_{n-1}\in S\Rightarrow r_0\in T$.
	\item $S=\emptyset\Rightarrow T=\emptyset$.
\end{itemize}
These constraints corresponds to the properties \textit{P1}, \textit{P2} and \textit{P3} in the peculiar case where $\alpha=3$. So the number of states verifying these constraints (valid states) is equal to $\#\mathcal T_3$, the value computed in Example \ref{exampleJG}. We prove all these states are both accessible and pairwise non-equivalent. %, that is their number reaches the bound $\#\mathcal T_3$ and the triple $A_1, A_2, A_3$ is a witness for double catenation.
\begin{proposition}\label{3-accessibility}
Any valid state $s=(p_i,S,T)$ of $A$ is accessible.
\end{proposition}

\begin{proof}
By induction over the size of $S\cup T$. First, any state of the form $(p_i,\emptyset,\emptyset)$ is accessible by $b^i$ from the initial state $(p_0,\emptyset,\emptyset)$. Next, consider some integer $\theta$ and suppose any state $(p_i,S,T)$ with $|S\cup T|\leq\theta$ is accessible. We prove, by cases, that any state $(p_i,S,T)$ with $|S\cup T|=\theta+1$ is also accessible.
\begin{enumerate}
	\item $i=m-1$
		\begin{enumerate}
			\item $S=\{q_0\}$ and $T=\emptyset$\\
						$(p_{m-2},\emptyset,\emptyset)\xrightarrow{b}s$
			\item $S=\{q_0\}$ and $T\neq\emptyset$
			
			\centerline{$\begin{array}{ll}
						(p_{m-1},\{q_0\},a^{n-1+k_0}\cdot (T\setminus\{r_{k_0}\}))&\xrightarrow a\\%{a(ac)^{n-2+k_0+p}}&=\\
						(p_{m-1},\{q_0,q_1\},a^{n-2+k_0}\cdot (T\setminus\{r_{k_0}\}))&\xrightarrow {(ac)^{n-2}}\\%{(ac)^{n-2+k_0+p}}&=\\
						(p_{m-1},\{q_0,q_{n-1}\},a^{k_0}\cdot T)&\xrightarrow {(ac)^{k_0+p}}s\\
			\end{array}$}
						
						The suffix $(ac)^{k_0}$ is sufficient in general, except when $k_0=0$, since one last occurrence of $ac$ is 				
						necessary over $A_2$.
			\item $|S|>1$\\
			If $j_1 > 1$ then
			
			\centerline{$\begin{array}{ll}
						(p_{m-1},a^{j_1}.(S\setminus\{q_0\}),a^{j_1}\cdot T)&\xrightarrow a \\% \\{a(ac)^{j_1-1}}s$
						(p_{m-1},\{q_0\}\cup a^{j_1-1}\cdot (S\setminus\{q_0\}),a^{j_1-1}\cdot T)&\xrightarrow {ac} \\%{a(ac)^{j_1-1}}s$
						(p_{m-1},\{q_0\}\cup a^{j_1-2}\cdot(S\setminus\{q_0\}),a^{j_1-2}\cdot T)&\xrightarrow {(ac)^{j_1-2}} s\\
			\end{array}$}

\bigskip			
			if $j_1=1$ then 
			
			\centerline{$(p_{m-1},a\cdot (S\setminus\{q_0\}),a\cdot T)\xrightarrow {a} s$}
		\end{enumerate} 
	\item $i<m-1$
		\begin{enumerate}
			\item $S=\{q_{j_0}\}$ and $T=\emptyset$
			
			\centerline{$\begin{array}{ll}
						(p_{m-2},\emptyset,\emptyset))&\xrightarrow{b^{i+2}}\\
						(p_i,\{q_{(i+1) \bmod 2}\},\emptyset &\xrightarrow {c^2}\\
						(p_i,\{q_{j_0}\},\emptyset) &\xrightarrow {a^{j_0}} s\\
					\end{array}$}
						
			\item $S=\{q_{j_0}\}$ and $T\neq\emptyset$\\
			Let $\delta=(k_0-j_0) \bmod p$ and $i'=\left\{\begin{array}{ll}
											i &\mbox{if }i>1\mbox{ or }\delta + p\mbox{ is even}\\
											1-i &\mbox{otherwise}\\
											\end{array}
											\right.$
			
			\centerline{$\begin{array}{ll}
						(p_{i'},\{q_{n-2}\},a^{k_0+1}\cdot(T\setminus\{r_{k_0}\}))&\xrightarrow a\\ %
						(p_{i'},\{q_{n-1}\},\{r_0\} \cup a^{k_0}\cdot(T\setminus\{r_{k_0}\}))&\xrightarrow {(ac)^{\delta +p}}\\ %
						(p_i,\{q_0\},\{r_\delta\} \cup a^{(-j_0) \bmod p}\cdot(T\setminus\{r_{k_0}\}))&\xrightarrow {a^{j_0}} s\\
						\end{array}$}
												
				As previously, the factor $(ac)^p$ ensures one occurrence of $ac$ even when $(k_0-j_0)\bmod p=0$.
		  \item $|S|>1$ and $r_{j_0+1}\not\in T$
				\begin{enumerate}
					\item $j_1>j_0+1$
					
					\centerline{$\begin{array}{ll}
								(p_{m-2},a^{j_0}\cdot (S\setminus\{q_{j_0}\}),a^{j_0}\cdot T)&\xrightarrow{bb}\\
								(p_0,\{q_1\}\cup a^{j_0}\cdot (S\setminus\{q_{j_0}\}),a^{j_0}\cdot T)&\xrightarrow{b^i}\\
								(p_i,\{q_{(1+i)\bmod 2}\}\cup a^{j_0}\cdot (S\setminus\{q_{j_0}\}),a^{j_0}\cdot T)&\xrightarrow{cc}\\
								(p_i,\{q_0\}\cup a^{j_0}\cdot (S\setminus\{q_{j_0}\}),a^{j_0}\cdot T)&\xrightarrow{a^{j_0}}s
						\end{array}$}																
								
					\item $j_1=j_0+1$ (we have $q_0,q_1\in a^{j_0}\cdot S$)
					
					\centerline{$\begin{array}{ll}
								(p_{m-2},a^{j_0}\cdot (S\setminus\{q_{j_1}\}),a^{j_0}\cdot T)&\xrightarrow{bb}\\
								(p_{m-2},a^{j_0}\cdot S,a^{j_0}\cdot T)&\xrightarrow{b^i}\\
								(p_i,a^{j_0}\cdot S,a^{j_0}\cdot T)&\xrightarrow{a^{j_0}} s\\
 					\end{array}$}																

				\end{enumerate} 
			\item $|S|>1$ and $r_{j_0+1}\in T$\\
						If $S={ Q}_{A_2}$ then note that $r_0,r_1\in T$ and so: $(p_i,S,a\cdot (T\setminus\{r_0\})\xrightarrow{a}s$.\\
						If $S\neq{ Q}_{A_2}$ then we set $\Delta=min\{\delta>0|q_{j_0+\delta}\not\in S\}$. If $\Delta=1$ then we first notice that $q_{n-1}\not\in a^{j_{0}+2}\cdot S$ and $q_{n-1}\in a^{j_{0}+1}\cdot S$. Hence,
						 $(p_i,a^{j_0+2}\cdot S,a^{j_0+2}\cdot (T\setminus\{r_{j_0+1}\}))$ is valid and
						$$(p_i,a^{j_0+2}\cdot S,a^{j_0+2}\cdot (T\setminus\{r_{j_0+1}\}))\xrightarrow{a}(p_{i},a^{j_{0}+1}\cdot S,a^{j_{0}+1}\cdot T)$$
						because $r_{p-1}\not\in a^{j_{0}+2}\cdot (T\setminus\{r_{j_{0}+1}\})$. 
						Furthermore, since $j_{0}=\min\{j\mid q_{j}\in S\}$, we 
						have $$(p_{i},a^{j_{0}+1}\cdot S,a^{j_{0}+1}\cdot T)\xrightarrow{a^{j_{0}+1}}s.$$
						%$(p_i,a^{j_0+2}.S,a^{j_0+2}.(T\setminus\{r_{j_0+1}\}))\xrightarrow{aa^{j_0+1}}s$\\
						Now let us examine the case when $\Delta>1$ and  set $R=\{r_{j_0+2},\ldots ,r_{j_0+\Delta}\}$. 
						We consider two situations:
				\begin{enumerate}
					\item $R\cap T=\emptyset$.
						\begin{enumerate}
							\item $i\neq (\Delta-2)\bmod m$\\
							Since $q_{n-1}\in a^{j_{0}+\Delta}\cdot S$ we have
							$$(p_{i-(\Delta-1)},a^{j_0+\Delta+1}\cdot S,a^{j_0+\Delta+1}\cdot (T\setminus\{r_{j_0+1}\}))
										\xrightarrow{a}(p_{i-(\Delta-1)},a^{j_0+\Delta}\cdot S,\{r_{0}\}\cup a^{j_0+\Delta}\cdot (T\setminus\{r_{j_0+1}\})).$$
							But $q_{0}\in a^{j_{0}+\Delta-1}\cdot S$ so
							$$(p_{i-(\Delta-1)},a^{j_0+\Delta}\cdot S,\{r_{0}\}\cup a^{j_0+\Delta}\cdot (T\setminus\{r_{j_0+1}\}))
										\xrightarrow{abcc}(p_{i-(\Delta-2)},a^{j_0+\Delta-1}\cdot S,
										\{r_{0}\}\cup a^{j_0+\Delta-1}\cdot (T\setminus\{r_{j_0+1}\})).$$
							Furthermore
							$$(p_{i-(\Delta-2)},a^{j_0+\Delta-1}\cdot S,
										\{r_{0}\}\cup a^{j_0+\Delta-1}\cdot (T\setminus\{r_{j_0+1}\}))
										\xrightarrow{(ab)^{\Delta-3}}(p_{i-1},a^{j_0+2}\cdot S,
										\{r_{0}\}\cup a^{j_0+2}\cdot (T\setminus \{r_{j_{0}+1}\}).$$
							But $r_{p-1}\not\in a^{j_{0}+2}\cdot (T\setminus \{r_{j_{0}+1}\}))$ implies
							$$(p_{i-1},a^{j_0+2}\cdot S,
										\{r_{0}\}\cup a^{j_0+2}\cdot (T\setminus \{r_{j_{0}+1}\})
										\xrightarrow{ab}(p_{i},a^{j_0+1}\cdot S,
										a^{j_0+1}\cdot T).$$
							And finally,
							$$(p_{i},a^{j_0+1}\cdot S,
										a^{j_0+1}\cdot T)
										\xrightarrow{a^{j_{0}+1}}s.$$
							%%%%%%ù
							%
									%	$(p_{i-(\Delta-1)},a^{j_0+\Delta+1}.S,a^{j_0+\Delta+1}.(T\setminus\{r_{j_0+1}\}))
									%	\xrightarrow{aabcc(ab)^{\Delta-2}a^{j_0+1}}s$
							\item $i= (\Delta-2)\bmod m$\\
							We proceed in a very similar way excepting that we start on $p_{m-2}$ on the first automaton rather than $p_{i-(\Delta -1)}=p_{m-1}$ and we use a slightly different prefix
							$aabcbcc$ instead of $aabcc$:
										$$(p_{m-2},a^{j_0+\Delta+1}\cdot S,a^{j_0+\Delta+1}\cdot (T\setminus\{r_{j_0+1}\}))
										\xrightarrow{aabcbcc(ab)^{\Delta-2}a^{j_0+1}}s.$$
						\end{enumerate}
					\item $R\cap T\neq\emptyset$.\\
						We set $\Psi=min\{\psi>0|r_{j_0+1+\psi}\in T\}$. The proof goes as in the previous case (i) but we do not need
						to contract with $c$ on $A_{2}$. We construct the good words by deleting the letters $c$ and replacing $\Delta$ by $\Psi$. So the last two cases are
						\begin{enumerate}
							\item $i\neq (\Psi-2)\bmod m$\\
										$(p_{i-(\Psi-1)},a^{j_0+\Psi+1}\cdot S,a^{j_0+\Psi+1}\cdot (T\setminus\{r_{j_0+1}\})
										\xrightarrow{aab(ab)^{\Psi-2}a^{j_0+1}}s$
							\item $i=(\Psi-2)\bmod m$\\
										$(p_{m-2},a^{j_0+\Psi+1}\cdot S,a^{j_0+\Psi+1}\cdot (T\setminus\{r_{j_0+1}\})
										\xrightarrow{aabb(ab)^{\Psi-2}a^{j_0+1}}s$
						\end{enumerate}
					\end{enumerate}
	  \end{enumerate} 
\end{enumerate}
\end{proof}

\begin{proposition}\label{3-separation}
Any two distinct valid states, $s=(p_i,S,T)$ and $s'=(p_{i'},S',T')$ of $A$ are non-equivalent.
\end{proposition}

\begin{proof}
There are three cases to consider:
\begin{itemize}
	\item If $T\neq T'$, without loss of generality, let $r_k\in T\setminus T'$. Then $a^{p-1-k}$ sends $s$ to a final state and $s'$ to a non-final one.
	\item Now, suppose $T=T'$ and $S\neq S'$. Without loss of generality, let $q_j\in S\setminus S'$. First, read the word $a^{n-1-j}$. It sends $s$ to $s_1=(p_i,S_1,T_1)$ and $s'$ to $s_1'=(p_{i'},S_1',T_1')$ with $q_{n-1}\in S_1\setminus S_1'$. Then, read the word $b^{m-i'}$ to send $s_1$ to $s_2=(p_{i_2},S_2,T_2)$ and $s_1'$ to $s_2'=(p_0,S_2',T_2')$ with $q_{n-1}\in S_2\setminus S_2'$. Now, we set
\[ t=\left\{
\begin{array}{ll}
	(n-p+1)\bmod{n} & \mbox{ if }(n-p)\bmod{n}\neq 0,1\\
	2-(n-p)         & \mbox{ otherwise}
\end{array}
\right.\]
By reading the word $a^tb$, one sends $s_2$ to $s_3=(p_{i_3},S_3,T_3)$ and $s'$ to $s_3'=(p_1,S_3',T_3')$ with $q_{n-p}\in S_3\setminus S_3'$ and $r_1\not\in T_3\cup T_3'$. Last, by reading the word $a^{p-1}$, we send $s_3$ to $(p_{i_3},S_4,T_4)$ and $s_3'$ to $(p_1,S_4',T_4')$ with $r_0\in T_4\setminus T_4'$. So we have reduced this case to the previous one.
	\item Last, if $T=T'$, $S=S'$ and $i\neq i'$, without loss of generality, let us suppose  that $i>i'$. Then by reading the word $b^{m-1-i}ca^{n-1}$, we send $s$ to $(p_{m-1},S_1,T_1)$ and $s'$ to $(p_{i'},S_1',T_1')$ with $q_0\in S_1\setminus S_1'$. That is we have reduced this case to one of the two previous ones.
\end{itemize}
\end{proof}

It follows from Propositions \ref{3-accessibility} and \ref{3-separation} that:

\begin{theorem}
The triple of Brzozowski automata $A_1, A_2, A_3$ defined in Table \ref{three-letters witness} is a $3$-letters  witness for the catenation of three languages.
\end{theorem}

% !TEX root = main.tex
\section{Conclusion}

In this paper we have dramatically reduced the size of the alphabet needed to produce a family of witnesses for multiple catenation: $(\alpha+1)$-letters alphabet witness for catenation of $\alpha$ languages. We obtain this result by using Brzozowski DFAs, giving some new evidence of the fact that these tools seems a very good starting point to discover witnesses.
We also give a simple recursive formulae for the bound. Its effective computation gives rise to a combinatorial expression involving compositions which is an efficient alternative to the formulae given by Gao and Yu \cite{GY09} in the optimal case where automata have only one final state. \\

It remains, at least, two open problems:
\begin{enumerate} 
\item The proof of the conjecture  given in the last section where a $\alpha$-letters alphabet witnesses is given for catenation of $\alpha$ languages, but only validated  for $\alpha=2,3$.
\item The optimality of the size of the alphabet. Clearly, it is true when $\alpha=2$ but is it still true for greater values ?
\end{enumerate}

\bibliography{../COMMONTOOLS/biblio,../COMMONTOOLS/bibjg}

\end{document}